%% file: main.tex
\documentclass[11pt,a4paper]{article}

\usepackage{amsmath,amssymb,subfigure}
\usepackage{comment}
\usepackage[multiple]{footmisc}
\usepackage{amsthm}
\usepackage{hyperref}
\usepackage{mathtools}
\usepackage[shortlabels]{enumitem}
\usepackage{bbm}
\usepackage{graphicx,graphics,float}
\usepackage{subfigure}
\usepackage{natbib}
\usepackage[normalem]{ulem}
\usepackage[left=1in, right=1in, top=1in, bottom=1in, includefoot, headheight=15pt]{geometry}
\usepackage{stmaryrd}
 \usepackage{tikz}
 \usetikzlibrary{arrows,shapes,positioning,chains,fit,shapes,calc}

\usepackage{color}
\definecolor{grey}{rgb}{.25,.0,.25}
\definecolor{darkgreen}{rgb}{0, .5, 0}
\definecolor{darkred}{rgb}{.5, 0, 0}

\newcommand{\D}{\mathbf{D}}
\usepackage{setspace}

\theoremstyle{plain}
\newtheorem{theorem}{Theorem}[section]
\newtheorem*{mocktheorem*}{Mock Theorem}
\newtheorem{proposition}[theorem]{Proposition}
\newtheorem{corollary}[theorem]{Corollary} 
\newtheorem{assumption}[theorem]{Assumption} 
\newtheorem{lemma}[theorem]{Lemma} 
\newtheorem{example}[theorem]{Example}
\theoremstyle{definition} 

\usepackage{multicol}

\input{definitions.tex}

\renewcommand{\d}{{\partial}}

\newcommand{\St}{\mathbb{S}}
\newcommand{\W}{\mathbb{W}}

\usepackage{relsize}
\usepackage{bm}

\begin{document}
\singlespacing
\title{\LARGE\bf Optimal Support for Distressed Subsidiaries -- a Systemic Risk Perspective}
\author{Maxim Bichuch\thanks{Department of Mathematics,
SUNY at Buffalo
Buffalo, NY 14260. 
Work  is partially supported by NSF grant DMS-1736414. Research is partially supported by the Acheson J. Duncan Fund for the Advancement of Research in Statistics.
Email: {\tt mbichuch@buffalo.edu}.} \and Nils Detering\thanks{Department of Statistics and Applied Probability, University of California, Santa Barbara, CA 93106, USA. Email: {\tt detering@pstat.ucsb.edu}.}}

\maketitle
\begin{abstract}
We consider a network of bank holdings, where every holding has two subsidiaries of different types. A subsidiary can trade with another holding's subsidiary of the same type. Holdings support their subsidiaries up to a certain level when they would otherwise fail to honor their financial obligations. 
We investigate the spread of contagion in this banking network when the number of bank holdings is large, and find the final number of defaulted subsidiaries under different rules for the holding support. We also consider resilience of this multilayered network to small shocks. Our work sheds light onto the role that holding structures can play in the amplification of financial stress.  
We find that depending on the capitalization of the network, a holding structure can be beneficial as compared to smaller separated entities. In other instances, it can be harmful and actually increase contagion.
We illustrate our results in a numerical case study and also determine the optimal level of holding support from a regulator perspective. 

\end{abstract}

\section{Introduction}
Financial risk management used to be focused on an individual firm, and on a bilateral basis. However, after the financial crisis of 2007-2009, we have gained the understanding that this risk can spread akin to a virus through the entire banking system. And similar to a virus, the risk increases the more nodes succumb to the shock. This phenomenon is known as systemic risk -- the risk that a (small) shock that hits the system spreads throughout the system to a degree that it endangers the entire system. This spread is also referred to as financial contagion. There are two distinct ways for this risk to spread -- through local connections (for example when one bank cannot honor its obligations to another), and through global connection (for example when asset prices are impacted, as a result of liquidations). This paper will focus on the contagion spread through local connections.

One of the first papers on the spread of contagion through local connections is \cite{EN01}. In this paper the authors show existence of a clearing equilibrium payment vector in a network of banks connected by liabilities.  Numerous generalizations and follow up studies include e.g., \cite{ACP14,HK15,BEST04,ELS13,U11,G11,BBCH17}.  Generalizations such as fire sales and bankruptcy costs have been considered in  \cite{E07,RV13,EGJ14,GY14,AW_15,CCY16, E07,EGJ14,AW_15,GHM12,CFS05,NYYA07,arinaminpathy2012size,AFM13,CLY14,AW_15,AFM16,bichuch2019optimization,bichuch2022repo,bichuch2020systemic}).  We also refer to, e.g., \cite{AW_15,Staum,H16} for additional reviews of this literature.

Another stream of literature considers contagion effects in random networks (\cite{ar:gk10,Cont2016,detering2019managing,detering2020financial}, and \cite{deteringDefaultFireSales} for an extension incorporating fire sales). In these papers, asymptotic methods (as the number of banks increases to infinity) are facilitated to determine the final damage to the system after some initial shock hits some banks. From the perspective of random graph theory, the results on contagion in random networks are linked to a process called {\em bootstrap percolation}. This is a process which models the spread of some activity or infection within a graph via its edges. It has been analysed for Erdos-R\'enyi random graphs~\cite{janson2012} and random regular graphs~\cite{MR2283230}, and these results have later been extended to inhomogeneous graph in \cite{Amini2014Threshold, Amini2014, Detering2015a}, and to clustered graphs in \cite{https://doi.org/10.48550/arxiv.2205.14782}. 

Notably, in all these models the banks in the network have always been thought of as an indivisible atom. In other words, the entire bank must be in one state only (e.g. solvent or insolvent). We are not aware of any results where the banks are divisible and one subsidiary of a bank can be in a different solvency state than another subsidiary. In reality this is not the case. First, big international banks are usually big holding companies, with dozens if not hundreds of subsidiaries, that are divided by business lines and locations (e.g. JPMorgan Chase\footnote{\url{https://www.sec.gov/Archives/edgar/data/19617/000119312508043536/dex211.htm}} and AIG\footnote{\url{https://www.sec.gov/Archives/edgar/data/5272/000110465920023889/exhibit21.htm}}). Second, a financial loss---possibly amplified through contagion---does not necessarily affect the entire bank (at least not immediately), but often starts with one (or more) of its subsidiaries (see e.g. the classical example of AIG in \cite{Lewis2010big}). Contagion can then spread both between subsidiaries of different banks (e.g. through local connections), and between different subsidiaries of the same bank (e.g. through the holding). 

For corporations, becoming a holding with subsidiaries has several advantages: 
It allows them to defer taxable business income and use earnings for other business opportunities or, by channeling income to low tax countries, reduce taxes altogether. In other instances, it allows them to limit the spillover risk if one business line is in trouble. In this case, the holding acts as a pure shareholder and enjoys limited liability. Often bank holdings naturally divide into subsidiaries based on location (i.e. Europe, US) and/or business lines (i.e. broker-dealers, commercial bank, insurance, asset management). The precise judicial setup differs across different countries but two main types of subsidiaries are prevalent:
\begin{itemize}
\item {\em Type A:} The holding acts as an asset holder with full liability for the business activity of its subsidiary. This setup is usually accompanied by a profit transfer agreement ensuring fiscal unity. It is common in Europe; and is becoming more common in the US\footnote{\url{https://www.fdic.gov/regulations/reform/resplans/plans/boa-165-2107.pdf}}\footnote{\url{https://www.federalreserve.gov/econres/notes/feds-notes/foreign-banks-asset-reallocation-intermediate-holding-company-rule-of-2016-20210512.htm}}. For the holding, this structure has the advantage that profits from the subsidiary are often collected on a pretax basis and can be netted with losses from other subsidiaries or the holding. For creditors of the subsidiary, an obvious advantage is that holdings support their subsidiaries, and cover some of their losses. 
\item {\em Type B:} The holding acts as a shareholder of its subsidiaries without explicit liability guarantee. This is the standard holding structure in the US\footnote{\url{https://www.wsj.com/articles/paul-kupiec-and-peter-wallison-the-fdics-bank-holding-company-heist-1419292997}}. From a risk perspective this holding structure is  beneficial for the holdings, as the insolvency of a subsidiary does not directly impact other business lines of the same holding. It might promote greater risk taking, because losses not covered by equity are immediately borne by other market participants, and not by the holding first.
\end{itemize}
While for the corporations themselves, the advantages are imminent, whether holdings are beneficial from a society perspective is less clear. In particular it is not well understood how shocks propagate in a network of bank holdings, and what influence the holding structure has on the contagion process. In this paper we address this important research question and try to understand the influence of the holding structures on the propagation of shocks through the financial system. We are particularly interested in the effect that support of the holding for its distressed subsidiary has on contagion. Holdings might support their subsidiaries beyond their financial liability for several reasons. For example they might fear the reputation cost of an insolvent subsidiary, or the threat of a possible change in credit rating if a subsidiary defaults. In other cases, the holding might support its subsidiary because it still believes in its business model, and considers the solvency problems to be of temporary nature. 

We consider two holding types which mimic the situation $A$ and $B$ described above.  Depending on the holding type A or B, the possible financial support of the holding for one subsidiary has different implications for the holding's other subsidiary. For type $A$ it weakens the other subsidiary, while for type $B$ it does not. We then consider a contagion process that is triggered by a financial shock and propagates throughout the network according to rules determined by the holding type. 

We investigate this contagion mechanism in a banking network where the number of banks is large. Under both holding types $A$ and $B$, we find the final number of defaulted subsidiaries at the end of the cascade. We further investigate resilience of this multilayered network to small shocks. 
We find that for well capitalized systems the holding structure can be beneficial as compared to smaller separated entities. This is because the holding is able to support a distressed subsidiary with capital from a subsidiary that is better off, therefore reducing the probability of default of the distressed subsidiary. However, if this support goes too far, it can actually amplify the spread of contagion, because healthy parts of the network get weakened and trigger new rounds of feedback effects. In a numerical case study we determine the optimal support level from a regulator perspective. We do not consider additional avenues through which the systemic risk can be mitigated in a financial network, such as intervention by a central bank, raising capital on the market or a takeover by another bank. Instead, we leave that for future research.

We pursue our analysis in a random network setup that raises also some interesting new theoretical challenges. From a technical perspective our model setup provides a first instance of a multilayered random network in which contagion spreads through different channels and where interactions between layers is through the nodes. This complicates the analysis compared to a one layer network, because the state of each node now depends on its local connectivity in each layer. 

Our analysis includes a multi-layer analog to the classical {\em bootstrap percolation} process mentioned above.

The paper is structured as follows: In Section~\ref{model} we state our model and our main result regarding the final number of defaulted subsidiaries of different types. In Section~\ref{resillience} we derive results that allow us to classify networks according to their resilience to small shocks. We give an example of a system in which the default of one subsidiary immediately triggers the default of the holding's other subsidiary. Still this system is more resilient than its fully separated counterpart due to the fact that the default of both subsidiaries can be slightly deferred with some holding support.  Section~\ref{casestudy} contains our numerical case study in which we also determine the optimal holding support. We summarize our results in Section \ref{conclusion}. All proofs are in Section~\ref{sec:proofs}.

\section{Model and main result}\label{model}

We consider a network of $n$ banks and label each bank by an index $i \in [n]:= \{ 1,\dots ,n \}$. Each bank is structured as a holding with two subsidiaries of different types $1$ and $2$. We denote the type $1$ subsidiary of bank $i$ by $i_1$ and the  type $2$ subsidiary of bank $i$ by $i_2$. We denote the set of type $1$ subsidiaries by $[n_1]:= \{ 1_1,\dots ,n_1 \}$ and the set of type $2$ subsidiaries by $[n_2]:= \{ 1_2,\dots ,n_2 \}$.  
We assume that subsidiaries only trade with subsidiaries of the same type, i.e. subsidiary of type 1 of bank $i$ might trade with subsidiary of type $1$ of another bank $j$ but not with its subsidiary of type $2$. For simplicity, in what follows we will refer to subsidiary of type $l$, as simply subsidiary $l$.
The assumption then implies that there is no cross trading between subsidiaries of different types, and linkage is only through the holdings. 
We consider a network that describes these trading activities. In this toy model we assume that every loan is of equal unit size $1$. For this let $e^l_{i,j} \in \{0,1 \} ,i,j \in [n], i\neq j, l\in \{ 1,2 \}$ denote the exposure of subsidiary $l$ of bank $i$ towards subsidiary $l$ of bank $j$. We then build a network by drawing a directed link from subsidiary $l$ of bank $i$ to subsidiary $l$ of bank $j$ if $e^l_{i,j}=1$. We do not allow for self loops and multiple edges. 
For simplicity, we assume that the recovery rate is zero, so that if $e^l_{i,j}=1$ and subsidiary $l$ of bank $j$ defaults, then subsidiary $l$ of bank $i$ will incur a loss of $1$ on its loan. Therefore, we are looking at short- and medium-term effects of defaults resulting in illiquidity driven defaults. Ultimately, the loan holder might recover some money from the loan, but that is likely to take some time.

We denote by $c_{i,1}\in \mathbb{Z}$ and $c_{i,2}\in \mathbb{Z}$ the capital of subsidiary $1$ and $2$ of bank $i$. The total capital of the bank is then equal to $c_i:=c_{i,1}+c_{i,2}$. This is the initial capital structure of each bank holding. 
As mentioned in the introduction already, in this paper we consider two types of holding structure. Type A that mimics the situation where the holding acts as an asset holder of the assets of its subsidiaries, and type B where the holding acts as a shareholder of its subsidiaries. We assume in the following that all holdings are of the same type, either all of type $A$ or all of type $B$. Our results could however easily be extended to networks that consist of type $A$ and type $B$ holdings. In what follows, several quantities that are introduced depend on $J\in \{ A, B\}$, but in order to lighten the notation we make this dependency only specific when necessary. 

We assume that some subsidiaries are in default, i.e. the set $\mathcal{D}_0=\mathcal{S}_{1} \cup \mathcal{S}_{2}$ is not empty, where
$$\mathcal{S}_{l} = \{ i_l \in [n_l] \colon \mbox{$i_l$ is in default} \},~l=1,2.$$

 The initial defaults in $\mathcal{D}_0$ now trigger a cascade that then evolves in {\em generations} or {\em rounds}. 
 We denote by $c_{i,1}^{(k)}$ and $c_{i,2}^{(k)}$ the capital of subsidiary $1$ and $2$ of bank $i$ in round $k\ge0$, and set $(c_{i,1}^{(0)},c_{i,1}^{(0)}):= (c_{i,1},c_{i,2})$. In round $k$, the total capital of the bank is then equal to $
c_{i,1}^{(k)}+c_{i,2}^{(k)}$. 

Let $x \in \{ 0, -1,-2,\dots \} \cup \{ -\infty \} $ be a global support level. We assume that all holdings will support their subsidiaries up to the same fixed level $x$, when possible. For type $A$ holdings, the support for a subsidiary in distress has to come directly from the other subsidiary of the same holding. It can be thought of the holding reallocating capital (by moving some assets) from the healthy subsidiary to the other distressed subsidiary. For holding structures of type $B$, the support for a troubled subsidiary does not reduce the capital of the other subsidiary. Instead, this support can be thought of as coming from a borrowing transaction of the holding against the capital of the healthy subsidiary. In any case, the holding supports its subsidiaries up to the level $x$ if possible and gives up on them when faced with substantial losses larger than $x$. For holdings of type $A$, limiting the loss from a distressed subsidiary at $x>-\infty$ is only possible by either selling this subsidiary or by offloading the assets of this subsidiary into a bad bank construction. 

Mathematically, this means that for holding $i$, its subsidiary $i_l,$ $l\in \{1,2\}$ is in default in round $k$ if one of the following conditions hold:
\begin{itemize}
\item The capital of subsidiary $i_l$ is below the support level $x$: $c^{(k)}_{i,l}\leq x$.
\item {\bf For a holding structure of type $A$}: $c_{i,1}^{(k)}+c_{i,2}^{(k)}\leq 0$ and $c_{i,l}^{(k)}\leq -x$. In other words, if the capital of subsidiary $i_l$ is at most $-x$, and the entire holding is in default, then the holding is not able to support the distressed subsidiary. 
Furthermore, the default of the distressed subsidiary weakens the other subsidiary which causes it to default as well, as effectively its entire capital was used to support the distressed subsidiary, which has now defaulted.
This results in the default of both subsidiaries.\\
{\bf For a holding structure of type $B$}: $c_{i,1}^{(k)}+c_{i,2}^{(k)}\leq 0$ and $c_{i,l}^{(k)}\leq 0$ In other words, if the capital of subsidiary $i_l$ is non-positive, and the entire holding is in default, so the holding again is not able to support the distresses subsidiary. Since each subsidiary acts as a share holder, this results in default of the distressed subsidiary, but the stronger subsidiary may still survive.
\end{itemize}

In each round the contagion expands, because a defaulted subsidiary causes a loss of $1$ to every subsidiary it traded with and therefore reduces their capital by $1$. If subsidiary $i_l$ of bank $i$ defaults in round $k$, then $c_{j,l}^{(k+1)}=c_{j,l}^{(k)}-1$ for all $j\in \{ j\in [n] : e^l_{i,j}\neq 0 \}$ unless subsidiary $j_l$ had already defaulted. 
Updating the capitals of all subsidiaries ends the round and the next round then starts. Note that if round $k$ leads to holding defaults, which can happen if $c_{i,1}^{(k)}+c_{i,2}^{(k)}\leq 0$ for some $i$, then the capital of the subsidiary that has already defaulted is not changed, and so is the capital of the other subsidiary that might default even with positive capital.

In round $k$, the set of banks with defaulted type $l$ subsidiary is denoted by 
\begin{equation}\label{process:generations}
\mathcal{S}_{l,k} = \{ i_l \in [n_l] \colon \mbox{subsidiary $i_l$ is in default in round $k$} \},
\end{equation}
and $\mathcal{D}_k = \mathcal{S}_{1,k} \cup \mathcal{S}_{2,k}$ denotes the set of all defaulted subsidiaries at step $k$. This then leads to the two cascades $$\mathcal S_{l,0} \subset  \mathcal{S}_{l,1} \subset \dots,~l=1,2,$$ where both cascade processes are strongly coupled through the holdings. The cascade of all subsidiaries is given by $$\mathcal{D}_{0} \subset  \mathcal{D}_{1} \subset \dots.$$ This cascade stops after at most $2n$ iterations and we denote by $\mathcal{S}_{l,2n-1}$ the final sets of defaulted subsidiaries for $l\in \{ 1,2 \}$. Such a process is exemplified for a small network with $n=3$ and support level $x=-1$ in Figures ~\ref{fig2} and \ref{fig2.1} for holdings of types $A$ and $B$ respectively.

In the following we want to determine $n^{-1} \abs{\mathcal{S}_{l,2n-1}}$ in a random network setting. For this let us first stress that when a holding supports a subsidiary that has a non-positive capital, we formalize this support not by actually increasing the capital of the supported subsidiary in order to lift it above $0$, but by reducing the capital threshold at which it defaults by $1$. Similarly in case of the holding type $A$, the support for one subsidiary does not reduce the other subsidiary's capital but instead increases the capital level at which the other subsidiary would default by $1$. While both descriptions are equivalent and lead to the same defaulted subsidiaries, actually changing the capital to formalize the holding support is less convenient for the mathematical analysis. In fact, our approach of changing the default level has the advantage that in order to determine whether a subsidiary of a holding with capital structure $(c_{i,1}^{(k)},c_{i,2}^{(k)})$ has defaulted, one simply needs to check whether $(c_{i,1}^{(k)},c_{i,2}^{(k)})$ is in some subset of $\mathbb{Z}\times \mathbb{Z}$. This subset is the default region and it differs for types $A$ and $B$. We shall specify the default regions now. 

For simplicity we assume that all capitals are bounded by $R\in \mathbb{N}$, and because a subsidiary defaults always when the capital is less or equal than $x$, we do not need to consider any capitals strictly less than $x$. We therefore define $\mathbf{D}:= (\mathbb{Z} \cap  [x,R])\times (\mathbb{Z} \cap  [x,R])$ as the domain for the holding capital. Above consideration then leads to the default regions $\mathbf{D}_{A,1} , \mathbf{D}_{A,2}\subset  \mathbf{D}$ ($\mathbf{D}_{B,1},\mathbf{D}_{B,2}\subset  \mathbf{D} $ respectively) for subsidiary $1$ and subsidiary $2$  and for holding of type $A$, respectively $B$. They are explicitly given by
\begin{align}
& \mathbf{D}_{A,1} = \{ (j,k)\in \D \vert j \leq x \mbox{ or }  j+k\leq 0, j\le -x\},\\
& \mathbf{D}_{A,2} =\{ (j,k)\in \D \vert k \leq x \mbox{ or }  j+k\leq 0, k\le -x\}, \\
& \mathbf{D}_{B,1}= \{ (j,k)\in \D \vert j \leq x \mbox{ or }  j+k\leq 0, j \le 0 \},\\
& \mathbf{D}_{B,2}= \{ (j,k)\in \D \vert k \leq x \mbox{ or }  j+k\leq 0, k\leq 0  \}.
\end{align}
\begin{figure}[t]
\begin{minipage}[b]{1\linewidth}
\begin{center}
\begin{tikzpicture}[->,>=stealth',shorten >=1pt,auto,node distance=2.3cm,
                    semithick]
\tikzset{
mystyle/.style={
  circle,
  inner sep=0pt,
  text width=6mm,
  align=center,
  draw=blue,
  fill=white
  }
}

\tikzset{
mystylered/.style={
  circle,
  inner sep=0pt,
  text width=6mm,
  align=center,
  draw=red,
  fill=white
  }
}
  \node[label={[shift={(-1.,-1.0)}]{\small bank 1}},circle,mystylered,label=left:{}] (A)                    {$-1$};
    \node[circle,mystyle,label=left:{}] (A2)  [below=0.08cm of A]                   {$2$};
  \node[label={[shift={(.0,.1)}]{\small bank 2}},circle,mystyle,label=left:{}]         (B) [above right of=A] {$1$};
      \node[circle,mystyle,label=left:{}] (B2)  [below=0.08cm of B]                   {$3$};
  \node[label={[shift={(1.0,-1.0)}]{\small bank 3}},circle,mystyle,label=right:{}]         (C) [below right of=B] {$0$};
        \node[circle,mystyle,label=left:{}] (C2)  [below=0.08cm of C]                   {$1$};

  \path (A) edge (C)
  edge [bend left] (B)
            edge               (C);
  \path (C2) edge [bend left, darkgreen] (B2);

\end{tikzpicture}
\end{center}
\end{minipage}
\begin{minipage}[t]{0.5\linewidth}
\begin{tikzpicture}[->,>=stealth',shorten >=1pt,auto,node distance=2.3cm,
                    semithick]
\tikzset{
mystyle/.style={
  circle,
  inner sep=0pt,
  text width=6mm,
  align=center,
  draw=blue,
  fill=white
  }
}

\tikzset{
mystylered/.style={
  circle,
  inner sep=0pt,
  text width=6mm,
  align=center,
  draw=red,
  fill=white
  }
}
  \node[label={[shift={(-1.,-1.0)}]{\small bank 1}},circle,mystylered,label=left:{}] (A)                    {$-1$};
    \node[circle,mystyle,label=left:{}] (A2)  [below=0.08cm of A]                   {$2$};
  \node[label={[shift={(.0,.1)}]{\small bank 2}},circle,mystyle,label=left:{}]         (B) [above right of=A] {$0$};
      \node[circle,mystyle,label=left:{}] (B2)  [below=0.08cm of B]                   {$3$};
  \node[label={[shift={(1.0,-1.0)}]{\small bank 3}},circle,mystylered,label=right:{}]         (C) [below right of=B] {$-1$};
        \node[circle,mystylered,label=left:{}] (C2)  [below=0.08cm of C]                   {$1$};

  \path (A) edge (C)
  edge [bend left] (B)
            edge               (C);
  \path (C2) edge [bend left, darkgreen] (B2);

\end{tikzpicture}
\end{minipage}\hfill
\begin{minipage}[t]{0.5\linewidth}
\tikzset{
mystyle/.style={
  circle,
  inner sep=0pt,
  text width=6mm,
  align=center,
  draw=blue,
  fill=white
  }
}

\tikzset{
mystylered/.style={
  circle,
  inner sep=0pt,
  text width=6mm,
  align=center,
  draw=red,
  fill=white
  }
}
\begin{tikzpicture}[->,>=stealth',shorten >=1pt,auto,node distance=2.3cm,
                    semithick]

  \node[label={[shift={(-1.,-1.0)}]{\small bank 1}},circle,mystylered,label=left:{}] (A)                    {$-1$};
    \node[circle,mystyle,label=left:{}] (A2)  [below=0.08cm of A]                   {$2$};
  \node[label={[shift={(.0,.1)}]{\small bank 2}},circle,mystyle,label=left:{}]         (B) [above right of=A] {$0$};
      \node[circle,mystyle,label=left:{}] (B2)  [below=0.08cm of B]                   {$2$};
  \node[label={[shift={(1.0,-1.0)}]{\small bank 3}},circle,mystylered,label=right:{}]         (C) [below right of=B] {$-1$};
        \node[circle,mystylered,label=left:{}] (C2)  [below=0.08cm of C]                   {$1$};

  \path (A) edge (C)
  edge [bend left] (B)
            edge               (C);
  \path (C2) edge [bend left, darkgreen] (B2);

\end{tikzpicture}
\end{minipage}
\caption{\label{fig2} Illustration of an exemplary contagion process for a network with three holdings of type $A$ with a support level $x=-1$. The pair of two circles arranged vertically form a holding with the top circle denoting subsidiary $1$ and the bottom circle subsidiary $2$. The edges between type $1$ subsidiaries are drawn in black and the edges between type $2$ subsidiaries are drawn in green. The cascade start (round $0$) is shown in the center upper graph. Subsidiary $1$ of bank $1$ has defaulted because its capital has reached the maximal support level and therefore defaults although the holding capital is still positive. This leads to $\mathcal{S}_{1,0}=\{1_1\}$ and $\mathcal{S}_{2,0}=\emptyset $, which triggers a short cascade evolving in two rounds. In round $1$, shown in the lower left graph, subsidiary $1$ of bank $3$ looses $1$ unit in capital, so its capital then becomes $c_{3,1}^{(1)}=-1$, which causes the entire bank holding $3$ to default and the round ends with ($\mathcal{S}_{1,1}=\{1_1, 3_1 \}$ and $\mathcal{S}_{2,1}=\{ 3_2 \}$). The result of round $2$, which is pictured in the lower right graph, is that subsidiary $2$ of bank $2$ looses one unit of capital, and its terminal capital is $c_{2,1}^{(2)}=2$. The reduced capital of subsidiary $2$ of bank $2$ does not lead to any further defaults ($\mathcal{S}_{1,2}=\{1_1, 3_1 \}$ and $\mathcal{S}_{2,2}=\{ 3_2 \}$).
}
\end{figure}

\begin{figure}[t]

\begin{minipage}[t]{0.5\linewidth}
\begin{tikzpicture}[->,>=stealth',shorten >=1pt,auto,node distance=2.3cm,
                    semithick]
\tikzset{
mystyle/.style={
  circle,
  inner sep=0pt,
  text width=6mm,
  align=center,
  draw=blue,
  fill=white
  }
}

\tikzset{
mystylered/.style={
  circle,
  inner sep=0pt,
  text width=6mm,
  align=center,
  draw=red,
  fill=white
  }
}
  \node[label={[shift={(-1.,-1.0)}]{\small bank 1}},circle,mystylered,label=left:{}] (A)                    {$-1$};
    \node[circle,mystyle,label=left:{}] (A2)  [below=0.08cm of A]                   {$2$};
  \node[label={[shift={(.0,.1)}]{\small bank 2}},circle,mystyle,label=left:{}]         (B) [above right of=A] {$1$};
      \node[circle,mystyle,label=left:{}] (B2)  [below=0.08cm of B]                   {$3$};
  \node[label={[shift={(1.0,-1.0)}]{\small bank 3}},circle,mystyle,label=right:{}]         (C) [below right of=B] {$0$};
        \node[circle,mystyle,label=left:{}] (C2)  [below=0.08cm of C]                   {$1$};

  \path (A) edge (C)
  edge [bend left] (B)
            edge               (C);
  \path (C2) edge [bend left, darkgreen] (B2);

\end{tikzpicture}
\end{minipage}\hfill
\begin{minipage}[t]{0.5\linewidth}
\begin{tikzpicture}[->,>=stealth',shorten >=1pt,auto,node distance=2.3cm,
                    semithick]
\tikzset{
mystyle/.style={
  circle,
  inner sep=0pt,
  text width=6mm,
  align=center,
  draw=blue,
  fill=white
  }
}

\tikzset{
mystylered/.style={
  circle,
  inner sep=0pt,
  text width=6mm,
  align=center,
  draw=red,
  fill=white
  }
}
  \node[label={[shift={(-1.,-1.0)}]{\small bank 1}},circle,mystylered,label=left:{}] (A)                    {$-1$};
    \node[circle,mystyle,label=left:{}] (A2)  [below=0.08cm of A]                   {$2$};
  \node[label={[shift={(.0,.1)}]{\small bank 2}},circle,mystyle,label=left:{}]         (B) [above right of=A] {$0$};
      \node[circle,mystyle,label=left:{}] (B2)  [below=0.08cm of B]                   {$3$};
  \node[label={[shift={(1.0,-1.0)}]{\small bank 3}},circle,mystylered,label=right:{}]         (C) [below right of=B] {$-1$};
        \node[circle,mystyle,label=left:{}] (C2)  [below=0.08cm of C]                   {$1$};

  \path (A) edge (C)
  edge [bend left] (B)
            edge               (C);
  \path (C2) edge [bend left, darkgreen] (B2);

\end{tikzpicture}

\end{minipage}
\caption{\label{fig2.1} In the same setup as Figure \ref{fig2}, illustration of a contagion process for a network with three holdings of type $B$ with a support level $x=-1$. The pair of two circles arranged vertically form a holding with the top circle denoting subsidiary $1$ and the bottom circle subsidiary $2$. The edges between type $1$ subsidiaries are drawn in black and the edges between type $2$ subsidiaries are drawn in green. 
Round $0$ is the same as for type $A$ -- subsidiary $1$ of bank $1$ defaults, $\mathcal{S}_{1,0}=\{1_1\}$ and $\mathcal{S}_{2,0}=\emptyset $. The difference is in round $1$, where while subsidiary $1$ of bank $3$ still looses $1$ unit in capital, and its capital also becomes $c_{3,1}^{(1)}=-1$, this will cause the holding to let it subsidiary to default, while the capital of subsidiary $2$ will remain $c_{3,1}^{(1)}=1.$ Therefore, the cascade of defaults stops. 
}
\end{figure}

\begin{figure}[ht]
\subfigure[]{
\begin{tikzpicture}[thick, scale=0.45, 
elipset/.style={
    ellipse, 
    draw=black, 
    minimum height=1.5em, 
    text width=4em,
    text centered, 
    on chain},
    ellipselarge/.style={
    ellipse, 
    draw=black, 
    minimum height=3.5em, 
    text width=4em,
    text centered},
  fsnode/.style={},
  ssnode/.style={},
  every fit/.style={ellipse,draw,inner sep=30pt,text width=10pt},
  inner/.style={circle,draw=blue!50,fill=blue!20,inner sep=1pt,scale=0.05},
    innerdead/.style={circle,draw=red!50,fill=red!20,inner sep=1pt,scale=0.03},
  innerblank/.style={circle,draw=none,inner sep=1pt,scale=0.00},
]

\begin{scope}[transparency group]
\begin{scope}[blend mode=multiply]
\node (A) at ( -3, 0) {};
\node (B) at (-3, 3) {};
\node (C) at ( 0, 0) {};
\node (G) at ( -3, -3) {};
\node (F) at (3, -3) {};
            \fill [blue!20] (G.center) -- (B.center) -- (F.center);
            \node (D) at ( 0, 0) {};
\node (E) at (0, -3) {};

            \fill [red!20] (G.center) -- (B.center) -- (F.center);
\fill[blue!20] (-5,-3) rectangle (-3,8); 
\fill[red!20] (-3,-5) rectangle (8,-3); 
\fill[blue!20] (-5,-5) rectangle (3,-3); 
\fill[red!20] (-5,-5) rectangle (-3,3); 
\end{scope}
\end{scope}

\draw[step=1cm,gray,very thin] (-5,-5) grid (12,12);
\draw[thick,->] (-5,0) -- (10,0) node[anchor=north west] {$c_{i,1}$};
\draw[thick,->] (0,-5) -- (0,10) node[anchor=south east] {$c_{i,2}$};
\draw (8 cm,1pt) -- (8 cm,-1pt) node[anchor=north west] {$R$};
\draw (1pt, 8 cm) -- (-1pt, 8 cm) node[anchor=south east] {$R$};

\draw (-3 cm,1pt) -- (-3 cm,-1pt) node[anchor=north] {$x$};
\draw (1pt, -3 cm) -- (-1pt, -3 cm) node[anchor=east] {$x$};
\begin{scope}[blend mode=multiply]
\draw[blue, dashed, very thick] (-2 cm,8cm) -- (-2 cm,3cm);
\draw[blue, dashed, very thick] (-2 cm, 3cm) -- (4 cm,-3cm);
\draw[blue, dashed, very thick] (4 cm,-3cm) -- (4 cm,-5cm);

\draw[red, dashed, very thick] (8 cm,-2cm) -- (3cm,-2cm);
\draw[red, dashed, very thick] (3cm, -2cm) -- (-3 cm,4cm);
\draw[red, dashed, very thick] (-5 cm,4 cm) -- (-3 cm,4cm);
\end{scope}

\draw[grey, dashed, very thick] (-2 cm,8cm) -- (8cm,8cm);
\draw[grey, dashed, very thick] (8cm,8cm) -- (8cm,-2cm);

\end{tikzpicture}}
   \hfill \subfigure[]{
\begin{tikzpicture}[thick, scale=0.45,
elipset/.style={
    ellipse, 
    draw=black, 
    minimum height=1.5em, 
    text width=4em,
    text centered, 
    on chain},
    ellipselarge/.style={
    ellipse, 
    draw=black, 
    minimum height=3.5em, 
    text width=4em,
    text centered},
  fsnode/.style={},
  ssnode/.style={},
  every fit/.style={ellipse,draw,inner sep=30pt,text width=10pt},
  inner/.style={circle,draw=blue!50,fill=blue!20,inner sep=1pt,scale=0.05},
    innerdead/.style={circle,draw=red!50,fill=red!20,inner sep=1pt,scale=0.03},
  innerblank/.style={circle,draw=none,inner sep=1pt,scale=0.00},
]

\begin{scope}[transparency group]
\begin{scope}[blend mode=multiply]
\node (A) at ( -3, 0) {};
\node (B) at (-3, 3) {};
\node (C) at ( 0, 0) {};
            \fill [blue!20] (A.center) -- (B.center) -- (C.center);
            \node (D) at ( 0, 0) {};
\node (E) at (0, -3) {};
\node (F) at (3, -3) {};
            \fill [red!20] (D.center) -- (E.center) -- (F.center);
\fill[blue!20] (-5,0) rectangle (-3,8); 
\fill[red!20] (0,-5) rectangle (8,-3); 
\fill[blue!20] (-5,-5) rectangle (0,0); 
\fill[red!20] (-5,-5) rectangle (0,0); 
\end{scope}
\end{scope}

\draw[step=1cm,gray,very thin] (-5,-5) grid (12,12);
\draw[thick,->] (-5,0) -- (10,0) node[anchor=north west] {$c_{i,1}$};
\draw[thick,->] (0,-5) -- (0,10) node[anchor=south east] {$c_{i,2}$};
\draw (8 cm,1pt) -- (8 cm,-1pt) node[anchor=north west] {$R$};
\draw (1pt, 8 cm) -- (-1pt, 8 cm) node[anchor=south east] {$R$};

\draw (-3 cm,1pt) -- (-3 cm,-1pt) node[anchor=north] {$x$};
\draw (1pt, -3 cm) -- (-1pt, -3 cm) node[anchor=east] {$x$};
\begin{scope}[blend mode=multiply]
\draw[blue, dashed, very thick] (-2 cm,8cm) -- (-2 cm,3cm);
\draw[blue, dashed, very thick] (-2 cm, 3cm) -- (1 cm,0);
\draw[blue, dashed, very thick] (1 cm,0) -- (1 cm,-3cm);

\draw[red, dashed, very thick] (8 cm,-2cm) -- (3cm,-2cm);
\draw[red, dashed, very thick] (3cm, -2cm) -- (0 cm,1cm);
\draw[red, dashed, very thick] (0 cm,1 cm) -- (-3 cm,1cm);
\end{scope}

\draw[grey, dashed, very thick] (-2 cm,8cm) -- (8cm,8cm);
\draw[grey, dashed, very thick] (8cm,8cm) -- (8cm,-2cm);
\end{tikzpicture}}

\caption{\label{fig1} (a) Illustration for type $A$ holding. (b) illustration for type $B$ holding. Blue filled area is the default region of subsidiary 1, red filled area is the default region of subsidiary 2. The blue and red dotted lines correspond to the boundaries $\d \mathbf{D}_{A,1}$ and $\d \mathbf{D}_{A,2}$ for type $A$ (respectively $\d \mathbf{D}_{B,1}$ and $\d \mathbf{D}_{B,2}$ for type $B$.)}
\end{figure}

With the help of these regions we may then rewrite 
$$\mathcal{S}_{l,k} = \{ i_l \in [n_l] \colon (c_{i,1}^{(k)},c_{i,1}^{(k)})\in  \mathbf{D}_{J,l} \},$$
for $J \in \{A, B\}$. 
Recall that, to simplify the notation, the dependency of $\mathcal{S}_{l,k}$ on $J$ is implicit.
In Figure~\ref{fig1} we display the default regions for both holding structures of type $A$ and $B$. By the definition of the contagion process for type $B$ it follows trivially that smaller $x$ leads to less infections. However, we will see in the following that the situation for type $A$ is more complicated and interesting phenomena arise. 

In order to determine $n^{-1} \abs{\mathcal{S}_{l,2n-1}}$ we need to specify our  random network setting. We fix $p_1\in \mathbb{R}_{+}$ and $p_2\in \mathbb{R}_{+}$ and assume that $e^l_{i,j} $ is present with probability $p_l/n$ for $l\in \{ 1,2 \}$ and that all edges are mutually independent. Instead of a fixed network we analyse a sequence of networks of increasing size. For $n\in\mathbb{N}$, let as before $c_{i,1}=c_{i,1}(n)$ and $c_{i,2}=c_{i,2}(n)$ be the initial capital of subsidiary $1$ and $2$ of bank $i$. Let then $\mathbf{c}_{1}(n) =(c_{1,1} (n), \dots , c_{n,1}(n) )$ and $\mathbf{c}_{2}(n) =(c_{1,2} (n), \dots , c_{n,2}(n) )$ be the capitals of the subsidiaries for a network of size $n$. We will often drop the dependency on $n$ in the notation when it does not lead to confusion. 
Let $X^l_{i,j}$ be the indicator that a directed edge from subsidiary $l$ of bank $i$ to subsidiary $l$ of bank $j$ is present. Let $D^{l-}_i:=\sum_{j\in [n]\setminus \{ i\} } X^l_{j,i}$ and $D^{l+}_i:=\sum_{j\in [n]\setminus \{ i\} } X^l_{i,j}$ be the in-, respectively out-degree, of subsidiary $l$ of bank $i$. Clearly then $\E [ D^{l \pm}_i ] = p_l + o(1) $ and $D^{l \pm}_i \xrightarrow{p}{} \poi (p_l)$ as $n\rightarrow\infty$, where $\xrightarrow{p}{}$ denotes convergence in probability.

 In order to ensure that our asymptotic statements are relevant also for networks of finite size we shall need some regularity. In particular we require that the proportion of banks with a certain capital structure stabilizes as the network size increases. This is ensured by the following standing assumption for the rest of this paper:

\begin{assumption}\label{ass:regularity}
For each $n\in\N$, denote the joint empirical distribution function of $\mathbf{c}_1(n)$ and $\mathbf{c}_2 (n)$ by
\[ F_n(x,y)=n^{-1}\sum_{i\in[n]}\1_{\{c_{i,1} (n)\leq x,c_{i,2}(n)\leq y\}},\quad  (x,y)\in \mathbf{D} . \]
We assume that there exists a distribution~$F$ on $\mathbf{D}$ such that it holds $\lim_{n\to\infty}F_n(x,y)=F(x,y)$ for all $(x,y)\in \mathbf{D}$. Denote by $(C_1,C_2)$ a random vector distributed according to $F$.
\end{assumption}
Considering a network of holding types $J \in \{A,B \}$, we note that for fixed $n\in \mathbb{N}$ the contagion process only starts if there exists at least one bank $i$ such that $(c_{i,1},c_{i,2} ) \in \mathbf{D}_{J,1} \cup \mathbf{D}_{J,2}$. 
Under Assumption~\ref{ass:regularity}, this will be the case for $n$ large if $\mathbb{P} \left( ( C_1,C_2) \in \mathbf{D}_{J,1} \cup \mathbf{D}_{J,2} \right)>0$. 

Denote with superscript $c$ the complement of the domain with respect to $\D$, e.g., $\D^c_{A,1} = \D \backslash \D_{A,1}$. We then further define boundaries of the default regions by:
\begin{align}
&\d \mathbf{D}_{A,1} = \{ (j,k)\in \D^c_{A,1} \vert j=x+1 \mbox{ or }  j+k=1, x<j\le -x\mbox{ or }  j=-x+1,k\le x\},\\
& \d\mathbf{D}_{A,2} =\{ (j,k)\in \D^c_{A,2}  \vert k=x+1 \mbox{ or }  j+k=1, x<k\le -x\mbox{ or }  k=-x+1,j\le x\}, \\
& \d\mathbf{D}_{B,1}= \{ (j,k)\in \D^c_{B,1} \vert j=x+1 \mbox{ or }  j+k = 1, j \leq 0 \mbox{ or } j=1,k\le0 \},\\
& \d \mathbf{D}_{B,2}= \{ (j,k)\in \D^c_{B,2} \vert k=x+1 \mbox{ or }  j+k =1, k\leq 0 \mbox{ or } k=1,j \le 0 \}.
\end{align}
For example for a holding to be in $\d \mathbf{D}_{A,1}$ means that its subsidiary $1$ is at risk of default. This is the case either because the subsidiary defaults as soon as it has a new defaulted neighbor ($j=x+1$ or $j+k=1,~ x<j\le -x$) itself, or because it defaults as soon as the other subsidiary has a new defaulted neighbor ($j+k=1,x<j\le -x$). We refer again to Figure~\ref{fig1} for an illustration of these sets. 

Moreover, we define the following subsets $\d \mathbf{D}_{J^{lm}} $ to be the set of capital structures such that a link to subsidiary $m$ results in a default of subsidiary $l$, for $J\in \{ A,B\}, l,m \in \{1,2\}$. These sets are given by
\begin{align}
&\d \mathbf{D}_{A^{11}} = \{ (j,k)\in \D^c_{A,1} \vert j=x+1 \mbox{ or }   j+k=1, x<j\le -x \mbox{ or } j=-x+1,k\le x\},\\
& \d\mathbf{D}_{A^{21}} =\{ (j,k)\in \D^c_{A,2}  \vert    j+k=1,k\le -x\}, \\
& \d\mathbf{D}_{B^{11}}= \{ (j,k)\in \D^c_{B,1} \vert j=x+1 \mbox{ or } j=1,k\le0 \mbox{ or } j+k = 1, j \leq 0\},\\
& \d \mathbf{D}_{B^{21}}= \{ (j,k)\in \D^c_{B,2} \vert j+k =1,k\leq 0  \},
\end{align}
and
\begin{align}
&\d \mathbf{D}_{A^{12}} = \{ (j,k)\in \D^c_{A,1} \vert j+k=1, j\le-x\},\\
& \d\mathbf{D}_{A^{22}} =\{ (j,k)\in \D^c_{A,2}  \vert k=x+1, x<k\le -x \mbox{ or } j+k=1 \mbox{ or } k=-x+1,j\le x\}, \\
& \d\mathbf{D}_{B^{12}}= \{ (j,k)\in \D^c_{B,1} \vert j+k = 1, j\leq 0 \},\\
& \d \mathbf{D}_{B^{22}}= \{ (j,k)\in \D^c_{B,2} \vert k=x+1 \mbox{ or } k=1,j\le0 \mbox{ or } j+k = 1, k \leq 0\}.
\end{align}
 
It clearly holds that $\d \mathbf{D}_{A,l} = \d \mathbf{D}_{A^{l 1}}\cup \d \mathbf{D}_{A^{l 2}}$ and $\d \mathbf{D}_{B,l} = \d \mathbf{D}_{B^{l 1}}\cup \d \mathbf{D}_{B^{l 2}}$ for $l\in \{1 ,2 \}$. 
These boundary regions turn out to be important as they allow us to derive criteria for the contagion process to stop. Heuristically, bank holdings with capital on the boundary have very vulnerable subsidiaries and if there are too many of such holdings, the contagion process can quickly regain momentum even after it had slowed down.

Our first result, Theorem~\ref{thm:threshold:model}, describes the size of the final fraction of defaulted type $1$ and type $2$ subsidiaries. Its proof is based on a sequential reformulation of the process. This formulation leads to the same final outcome, but instead of exploring the effect of all defaulted subsidiaries at once, in each step only the effect of one defaulted subsidiary is considered. It turns out that in this reformulated process, the entire state of the system can be described in the limit ($n\rightarrow \infty$) by a system of continuous functions.  In the following we introduce these functions and provide some heuristic explanation for the results whose rigorous proof is given in Section~\ref{sec:proofs}. 

Let 
$\psi_i (\lambda):=e^{-\lambda }\frac{\lambda^i}{i !}$ for $\lambda \geq 0$ be the probability that a Poisson distributed random variable with parameter $\lambda$ takes value $i$. Let further $\Psi_k (\lambda) =\sum_{j=k}^{\infty} e^{-\lambda  }  \frac{\lambda^j}{j!}$ for $\lambda \geq 0$ be the probability that a Poisson distributed random variable with parameter $\lambda$ is greater or equal than $k$. We further define $\Psi_k (\lambda)=1$ for $k<0$.
We also define $\psi_{i,j} (\lambda_1,\lambda_2):= e^{-\lambda_1 }\frac{\lambda_1^i}{i !}e^{-\lambda_2 }\frac{\lambda_2^j}{j !}$ for $\lambda_1,\lambda_2\geq 0$ to be the probability that two independent Poisson distributed random variables $X$ and $Y$ with parameter $\lambda_1$ and $\lambda_2$ respectively, take the values $i$ and $j$, i.e. $\P (X=i,Y=j)=\psi_{i,j} (\lambda_1,\lambda_2)$.

Now define the functions $a_{j,k}(z_1,z_2)$ for $(j,k)\in \mathbf{D}$ by
\begin{equation}
a_{j,k}(z_1,z_2) = \E [  \1_{\{ (C_1,C_2) \geq (j,k) \}} \psi_{C_1-j,C_2-k} (p_1 z_1,p_2 z_2) ],
\label{eq:a-jk}
\end{equation}
where the inequality $(C_1,C_2) \geq (j,k)$ is meant component-wise. In the continuous approximation mentioned above, the two arguments $z_1$ and $z_2$ describe the current number of defaulted type $1$ respectively type $2$ subsidiaries whose effect on the system has already been taken into account, divided by $n$. These get updated in subsequent rounds. For a holding with original subsidiary capitals $(C_1,C_2) \geq (j,k)$, the probability that due to links from these defaulted institutions, the capital has been reduced from $(C_1, C_2)$ by $(C_1-j,C_2-k)$ to $(j,k)$, is approximately given by $\psi_{C_1-j,C_2-k} (p_1 z_1,p_2 z_2)$, and thus heuristically $a_{j,k}(z_1,z_2)$ describes the fraction of holdings that have the capital $(j,k)$, after a fraction $z_1$ (respectively $z_2$) of type $1$ (respectively type $2$) subsidiaries have defaulted.  

With the functions $\{a_{j,k}\}_{(j,k)\in \mathbf{D}}$ defined, we may sum over all those $(j,k)$ that are in the default region of subsidiary $l$ in order to obtain the total fraction of defaulted type $l$ subsidiaries $\sum_{(j,k) \in \mathbf{D}_{J,l}} a_{j,k}(z_1,z_2) $. The effect of the fraction $z_l$ of them on the system has already been considered, and it remains to explore the impact of the rest (if positive), which is given by 
\begin{eqnarray}
f_{l}(z_1,z_2)&:=& \sum_{(j,k) \in \mathbf{D}_{J,l}} a_{j,k}(z_1,z_2)  -z_l,
\label{eq:f-a}
\end{eqnarray}
for $J\in \{A,B \}, l\in \{1,2 \}$. Recall that to simplify the notation, the dependency on $J$ of $f_l(z_1, z_2))$ is implicit.

A calculation based on Figure~\ref{fig1} concludes that 
\begin{eqnarray}
f^A_{1} (z_1,z_2)&=&   -z_1  + \P ((C_1,C_2) \in \mathbf{D}_{A,1} ) +\E [ \1_{\{ (C_1,C_2) \in \mathbf{D}_{A,1}^c \}} \Psi_{C_1-x} (p_1 z_1)] \nonumber \\
&+& \E \left[   \sum_{j=(x+1)}^{C_1\wedge -x} \1_{\{ (C_1,C_2) \in \mathbf{D}_{A,1}^c \}}   \psi_{C_1- j} (p_1 z_1 )  \Psi_{C_2+ j } (p_2 z_2) \right], \nonumber 
\end{eqnarray}
and
\begin{eqnarray}
f^B_{1} (z_1,z_2)&=&   -z_1  + \P ((C_1,C_2) \in \mathbf{D}_{B,1} ) +\E [ \1_{\{ (C_1,C_2) \in \mathbf{D}_{B,1}^c \}} \Psi_{C_2-x} (p_2 z_2)] \nonumber \\
&+& \E \left[  \sum_{j=x+1}^{C_1\wedge 0} \1_{\{ (C_1,C_2) \in \mathbf{D}_{A,1}^c \}}   \psi_{C_1- j} (p_1 z_1 )  \Psi_{C_2+ j } (p_2 z_2) \right]. \nonumber 
\end{eqnarray}
The corresponding expressions $f^A_{2}$ and $f^B_{2}$ for subsidiary $2$ can be derived similarly.

It will turn out that the first joint zero of the functions $f_1$ and $f_2$ allows us to determine the final number of defaulted type $1$ and type $2$ subsidiaries. 
Note that for $x=0$ we get that $f^A_{1} (z_1,z_2) = f^B_{1} (z_1,z_2)=  \P ((C_1,C_2) \in \mathbf{D}_1) -z_1  + \E [ \1_{\{ (C_1,C_2) \in \mathbf{D}^c_1\}} \Psi_{C_1} (p_1 z_1)] $ is constant in $z_2$ and $f^A_{2} (z_1,z_2) = f^B_{2} (z_1,z_2)$ is constant in $z_1$ and thus the systems decouple. In this case the functions $f_1$ and $f_2$ describe contagion processes in two completely separated networks. 

Let us also calculate the partial derivatives of $f_l$ as they appear in the statement of the following theorem. With $\bm{z} = (z_1,z_2)$, for $J\in \{A,B \}, l,m\in \{1,2 \}$ one obtains that
\begin{eqnarray}\label{derivative}
\frac{\partial f_l}{\partial z_m} (\bm{z}) &=& - \delta_{l,m}   +p_m \sum_{(j,k) \in  \d \mathbf{D} _{J^{lm}}  }  a_{j,k}(\bm{z}), \label{deriv:f}
\end{eqnarray}
where $\delta_{i,j} = \1_{\{ i=j\}}$.
For a vector $\bm{v} = (v_1, v_2)\in \mathbb{R}^2$ we obtain the directional derivatives 
\begin{eqnarray}\label{directional:der}
D_{\bm{v}} f_l (\bm{z})  &=&  - v_l   + v_1 p_1 \sum_{(j,k) \in  \d \mathbf{D} _{J^{l1}}  }  a_{j,k}(\bm{z}) + v_2 p_2 \sum_{(j,k) \in  \d \mathbf{D} _{J^{l2}}  }  a_{j,k}(\bm{z}).
\end{eqnarray}
In the following proposition we collect some more properties of the function $f_l$ that will be important for the statement and proof of our first main Theorem~\ref{thm:threshold:model}. 
\begin{proposition}\label{properties:f}
Let $J\in \{A,B \}$. The functions $f_l, l \in \{ 1,2 \}$ are continuous and:
\begin{enumerate}
    \item $f_{l}(0,0)=\P((C_1,C_2)\in \mathbf{D}_{J,l})$ and thus $(0,0)$ is a joint zero of the functions $f_{1},f_{2}$ if and only if $\P((C_1,C_2)\in \mathbf{D}_{J,1}\cup \mathbf{D}_{J,2})=0$.
    \item $f_1$ is increasing in its second argument and $f_2$ is increasing in its first argument.
\end{enumerate}

\end{proposition}

We are now ready to state the first result that allows us to determine the final number of defaulted type $1$ and type $2$ subsidiaries. We provide several applications of this result in Section~\ref{casestudy}.
\begin{theorem}\label{thm:threshold:model}
Let $J\in \{A,B \}$. Consider a sequence of financial systems satisfying Assumption \ref{ass:regularity} and let $\mathbb{P} \left( (C_1 , C_2) \in \mathbf{D}_{J,1}\cup \mathbf{D}_{J,2} \right)>0 $. Then there exists a unique smallest (component-wise) positive joint root  $\hat{z}=(\hat{z}_1,\hat{z}_2)$ of the functions $f_{i}(z_1,z_2), i\in \{ 1,2\}$.
Moreover, the following holds for $i\in \{1,2\}$:
\begin{enumerate}
\item For all $\epsilon>0$, with high probability
$ n^{-1}\mathcal{S}_{i,2n-1} \geq \hat{z}_i - \epsilon.$
\item If in addition there exists a vector $\bf{v} \in \mathbb{R}^2_+$ such that 
$D_{\bf{v}} f_1 (\hat{z}), D_{\bf{v}} f_2 (\hat{z})  <0$, then
\[ n^{-1}\mathcal{S}_{l,2n-1} \xrightarrow{p}  \hat{z}_i,\quad\text{as }n\to\infty. \] Here $ \xrightarrow{p}$ denotes convergence in probability.
\end{enumerate}
\end{theorem}
These results are also intuitive. According to (\ref{eq:f-a}), the functions $f_{1}(z_1,z_2)$ and $f_{2}(z_1,z_2)$ quantify in a sequential exploration the fraction of defaulted subsidiaries of types $1$ and $2$ respectively, that still need to be explored. Therefore, for any $(z_1,z_2)<(\hat{z}_1,\hat{z}_2)$, there exist defaulted subsidiaries whose effect on the system has not been explored yet, and the first joint zero $(\hat{z}_1,\hat{z}_2)$ serves as a lower bound on the number of defaulted subsidiaries of type $1$ and $2$. 

The derivative condition 2. in the last theorem is a stability condition on the root. In a 1-dimensional setting it would simply corresponds to the derivative not being equal to zero in the root and would ensure that the root is not a saddle point. If this condition does not hold, then a minor modifications of the system (in terms of the specification of $(C_1,C_2)$) could change the outcome of the contagion process significantly. In particular the outcome might be influenced by a sub-linear proportion of the holdings/vertices. Knowing the limit $\lim_{n\to\infty}F_n(x,y)$ (see Assumption~\ref{ass:regularity}) is then not sufficient to determine the actual outcome of the process.

In this work we do not aim for the greatest possible generality but for a flexible framework that allows to analyze the effect of subsidiaries and holdings in mitigating and propagating contagion effects. The model can be extended in several directions:
\begin{itemize}
\item One can consider a model in which the support level $x$ depends on the holding $i$. Under a convergence assumption similar to Assumption \ref{ass:regularity} but incorporating the individual support levels $x_i$, the system can then be described by a random vector $(C_1,C_2,X)$. The only change would be that in the specification of $f$ the $x$ is replaced by the random variable $X$.  In fact we will consider such an example in Section~\ref{casestudy}.
\item We could consider a model in which the holding itself has capital. The support level $x$ could then depend on this capital. This leads to a more complicated function $f$ and a more involved exploration process but one still obtains a $2$-dimensional fixed-point equation that allows one to determine the final number of defaults.
    \item One could also consider more subsidiaries. Considering $M\in \mathbb{N}$ instead of only $2$ subsidiaries would lead to an $M$ dimensional fixed-point equation. We would like to remark however that in some situations a setup with two subsidiaries might work well even for an empirical analysis. Say a bank has $M>2$ subsidiaries but the distress starts only in one subsidiary. As a first approximation, one could then consider the remaining $M-1$ subsidiaries as one entity which supports the distressed subsidiary up to a capital of $x$.
    \item More general default regions can be considered, allowing for different ways the distress spreads through the holding. This is possible as long as the default region has certain natural properties. For example the default of one subsidiary needs to be more likely with the capital of the other subsidiaries decreasing. Shocks to a subsidiary can not improve the situation of the other subsidiary of the same holding. 
    \item The multi-layer random network that we consider is very homogeneous in terms of the degrees of the subsidiaries. In fact, as $n\rightarrow \infty$, the degree $D^{l \pm}_i \xrightarrow{p}{} \poi (p_l)$. Since the Poisson distribution has very light tails, large degrees are very rare. More heterogeneity in terms of the connection probabilities can lead to larger degree heterogeneity. In order to archive this, one could for example assign four weights $w^-_{i,1},w^+_{i,1}$ and $w^-_{i,2},w^+_{i,2}$ to each holding $i$, and assume that the probability that subsidiary $l$ of holding $i$ has exposure to subsidiary $l$ of holding $j$ is given by $(w^-_{i,l}w^+_{j,l})/n$. Subsidiaries with large weights will then have a larger tendency to have exposures. 
    \item Trading between subsidiaries of different types. While we wanted to focus here sorely on the impact of the coupling through the holding, one could also allow for direct trading between subsidiaries of different types. 
\end{itemize}

\section{Resilience}\label{resillience}
In this section we consider how small infections can propagate through a financial network and amplify to lead to significant damage. For this we start with an a-priory uninfected network, parameterized by $(C_1,C_2)$ with $\P((C_1,C_2)\in \mathbf{D}_{J,1}\cup \mathbf{D}_{J,2})=0$ and then apply ex-post infections. The ex-post infections will lead to $(\tilde{C}_1,\tilde{C}_2)$, a network with lower capitalized banks, and we determine the spread of contagion in this new network based on the characteristics of $(C_1,C_2)$. This provides a measure of resilience of the system $(C_1,C_2)$ to shocks.

Since our results hold for both, type $A$ and type $B$ holdings, and the proofs are generic, we drop the holding type $J \in \{A, B \}$ in the notation, and since
we will work with different financial systems we now include the random variable $(C_1,C_2)$ in the notation of the functions $f_l$ defined in (\ref{eq:f-a}) and write $f_l(C_1,C_2;z_1,z_2),~l\in\{1,2\}$. 
Before we move on, we collect some additional properties of these functions which we need in the following.

\begin{proposition}\label{representationf:P}
Let $(C_1,C_2)$ be a configuration for the financial system and let $f_l(C_1,C_2;\cdot), l\in \{1,2 \}$ be its functional. Let further $\Lambda^l_{n,h} (z_1 p_1 , z_2 p_2 ):= \P ((n,h)- (X_1(p_1 z_1),X_2(p_2 z_2)) \in \mathbf{D}_{J,l})$ with $X_1$ and $X_2$ mutually independent Poisson distributed random variables with parameter $p_1 z_1$ and $p_2 z_2$ respectively. Then, 
\begin{equation}\label{representation:f}
f_l(C_1,C_2;z_1,z_2)= \E [\Lambda^l_{C_1,C_2} (z_1 p_1 , z_2 p_2 )] -z_l.
\end{equation}
Moreover, by the law of total probability
\begin{equation}\label{representation:f:P}
f_l(C_1,C_2;z_1,z_2)= \P \left( (C_1,C_2)-(X_1(p_1 z_1),X_2(p_2 z_2)) \in \mathbf{D}_{J,l}\right)-z_l,
   \end{equation}
    with $X_1(p_1 z_1)$ and $X_2(p_2 z_2)$ mutually independent and independent of $(C_1,C_2)$ Poisson distributed random variables with parameter $p_1 z_1$ and $p_2 z_2$.
\end{proposition}
With the help of the previous proposition it becomes clear from the representation of $f_1,~ f_2$ that a larger default region leads to a larger number of defaulted subsidiaries. In particular in the case of type $B$ 
holdings, a decrease of $x$ leads to fewer infections. Of course this observation can also be made directly from the specification of the process in Section~\ref{model}, but in order to derive our results in this section the representations provided in Proposition~\ref{representationf:P} will turn out useful. 

As a direct consequence of the last proposition we obtain the following two corollaries:
\begin{corollary}\label{corollary:f:cont:in:C}
Let $\{(C^n_1,C^n_2)\}_{n=1}^\infty$ be a sequence of random variables on $\mathbf{D}$ that converges weakly to $(C_1,C_2)$. Then 
$f_l(C^n_1,C^n_2;\cdot,\cdot)$ and $D_{\bm{v}} f_l (C^n_1,C^n_2;\bm{z}) $ converge to $f_l(C_1,C_2;\cdot,\cdot)$ and $D_{\bm{v}} f_l (C_1,C_2;\bm{z}) $, $l\in\{1,2\}$  for all $\bm{v} \in \mathbb{R}^2$, uniformly on compacts. 
\end{corollary}

\begin{corollary}\label{monotonicity:f}
Let $(C_1,C_2)$ and $(\tilde{C}_1,\tilde{C}_2)$ be configurations for the financial system and let $f_l(C_1,C_2;\cdot)$ respective $f_l(\tilde{C}_1,\tilde{C}_2;z_1,z_2), l\in \{1,2 \}$ be their functional.
Then:
\begin{enumerate}
\item If $\P ((\tilde{C}_1,\tilde{C}_2) \leq (j,k))\geq \P ((C_1,C_2) \leq (j,k))$ for all $(j,k)\in \mathbf{D}$, that is  $(C_1,C_2)$ stochastically dominates $(\tilde{C}_1,\tilde{C}_2)$, then $$f_l(C_1,C_2;z_1,z_2) \leq f_l(\tilde{C}_1,\tilde{C}_2;z_1,z_2)$$ for all $z_1,z_2\geq 0$. 
\item Moreover, denote by $\mathbf{B}_{J,l}$ the set of entrance points to the default region $\mathbf{D}_{J,l}$, that is 
\begin{align}
& \mathbf{B}_{A,1} = \{ (j,k)\in \D_A \vert j = x,k\geq -x \mbox{ or }  j+k =  0, x \le k\le -x  \mbox{ or } j = -x, k\le x\},\\
& \mathbf{B}_{A,2} =\{ (j,k)\in \D_A \vert k = x,j\geq -x \mbox{ or }  j+k =  0, x \le j\le -x  \mbox{ or } j = -x, j\le x\},\\
& \mathbf{B}_{B,1}= \{ (j,k)\in \D_B \vert j = x,k\geq -x \mbox{ or }  j+k =  0, 0 \le k\le -x  \mbox{ or } j = 0, k\le 0\},\\
& \mathbf{B}_{B,2}= \{ (j,k)\in \D_B \vert k = x,k\geq -x \mbox{ or }  j+k =  0, 0 \le j\le -x  \mbox{ or } j = 0, j\le 0\}.
\end{align}
If in addition there exists $(j,k)\in \mathbf{D}^c_{J,l}\cup \mathbf{B}_{J,l}$ such that $\P ((\tilde{C}_1,\tilde{C}_2) \leq (j,k)) > \P ((C_1,C_2) \leq (j,k))$, then $$f_l(C_1,C_2;z_1,z_2) < f_l(\tilde{C}_1,\tilde{C}_2;z_1,z_2)$$ for all $z_1,z_2\geq 0$.
\end{enumerate}
\end{corollary}

We now move on to the statement of our main result in this section Theorem~\ref{resilience} and first specify what is exactly meant by a shock to the system. 

{\bf Specification of financial shock:} Following previous literature on contagion in financial networks \cite{Cont2016,detering2019managing,deteringDefaultFireSales,detering2020suffocating} we now consider a network that has a-priory no defaults, meaning $\P((C_1,C_2)\in \mathbf{D}_{J,1}\cup \mathbf{D}_{J,2})=0$. 
We then apply a small shock of size $\varepsilon$ to the system where $\varepsilon$ denotes the fraction of holdings affected by the shock. This shock has the effect that the capital of some subsidiaries is changed (i.e. $\P ((\tilde{C}_1,\tilde{C}_2) \neq(C_1,C_2))= \varepsilon >0$)), and we assume that:
\begin{enumerate}
    \item $\P ((\tilde{C}_1,\tilde{C}_2) \leq (C_1,C_2))=1$ (no one is better off after the shock),
\item $\P (\tilde{C}_1 ,\tilde{C}_2)\in \mathbf{D}_{J,1} \cup \mathbf{D}_{J,2} )>0$ (some defaults occur due to the shock).
\end{enumerate}
We call $(\tilde{C}_1,\tilde{C}_2)$ with the above properties the shocked system $(\tilde{C}_1,\tilde{C}_2)$ of $(C_1,C_2)$.

{\bf Resilience:} We now call a financial system described by $(C_1,C_2)$ {\em resilient} if for every given $\tilde{\varepsilon}>0$, there exists $\varepsilon$ such that for all shocked systems $(\tilde{C}_1,\tilde{C}_2)$ of $(C_1,C_2)$ with $\P ((\tilde{C}_1,\tilde{C}_2) \neq(C_1,C_2))\leq \varepsilon$, the number of defaulted subsidiaries $n^{-1}(\mathcal{S}_{1,2n-1}+\mathcal{S}_{2,2n-1})$ at the end of the default processes is less then $\tilde{\varepsilon}$ w.h.p.. 

If on the contrary, there exists a lower bound $\Delta$, such that for every shocked system $(\tilde{C}_1 ,\tilde{C}_2)$ of $(C_1,C_2)$, it holds that $n^{-1}(\mathcal{S}_{1,2n-1}+\mathcal{S}_{2,2n-1})\geq \Delta$ w.h.p., then we call the system {\em non-resilient}.

The following theorem gives criteria for resilience and non-resilience based on the derivative at ${\bf{0}}$ of the functional of the uninfected network, parametrized by $(C_1,C_2)$. 
\begin{theorem}\label{resilience}
We have the following statement regarding resilience:
\begin{enumerate}
    \item If there exists a vector $\bf{v} \in \mathbb{R}^2_+$ such that $D_{\bf{v}} f_1 (C_1,C_2;{\bf{0}}), D_{\bf{v}} f_2 (C_1,C_2;{\bf{0}})  <0$, then the network is resilient.
    \item If there exists a vector $\bf{v} \in \mathbb{R}^2_+$ such that $D_{\bf{v}} f_1 (C_1,C_2;{\bf{0}}), D_{\bf{v}} f_2 (C_1,C_2;{\bf{0}})  >0$, then the network is non-resilient.
\end{enumerate}
In particular the conditions 1. and 2. are mutually exclusive.
\end{theorem}
Let us provide some rough sketch of the proof of Theorem~\ref{resilience}, which is fully worked out in the Appendix: We know by Proposition~\ref{properties:f} that for the network without defaults it holds that $f_1 (C_1,C_2;{\bf{0}})=f_2 (C_1,C_2;{\bf{0}})=0$. If condition $1.$ holds and there exists a vector $\bf{v}$ such that the directional derivative of both functions $f_1(C_1,C_2,\cdot)$ and $f_2 (C_1,C_2,\cdot)$ is negative, then this implies that in the direction $\bf{v}$ both functions become negative. As long as the shock is small, i.e. $\P ((\tilde{C}_1,\tilde{C}_2) \neq(C_1,C_2))=\varepsilon $ small, we obtain by Corollary~\ref{corollary:f:cont:in:C} that $f_l (C_1,C_2;z_1,z_2 )$ is close to $f_l (\tilde{C}_1,\tilde{C}_2;z_1,z_2 )$  and $D_{\bf{v}} f_l (C_1,C_2;z_1,z_2)$ is close to $D_{\bf{v}} f_l (\tilde{C}_1,\tilde{C}_2;z_1,z_2)$ for $l\in \{1,2 \}$ and $(z_1,z_2)$ in some compact set. Despite $\max \{f_1 (\tilde{C}_1,\tilde{C}_2;{\bf{0}}),f_2 (\tilde{C}_1,\tilde{C}_2;{\bf{0}})  \}>0$ this will imply that for small $\varepsilon$ there exists a small $t_0 $ such that $f_l (C_1,C_2;t_0 \bf{v}) \leq 0$ for $l\in \{ 1, 2\}$. The following Lemma~\ref{upper:bound} then shows that $t_0 \bf{v}$ serves as an upper bound for the first joint zero of $f_1 (\tilde{C}_1,\tilde{C}_2;\cdot )$ and $f_2 (\tilde{C}_1,\tilde{C}_2;\cdot)$ which then leads to the conclusion of resilience. 
\begin{lemma}\label{upper:bound}
Let $\tilde{\bm{z}}= (\tilde{z}_1,\tilde{z}_2)$ be such that $(f_1(\tilde{\bm{z}}) , f_2(\tilde{\bm{z}}) ) \leq \bm{0}$, then it holds that $\hat{\bm{z}} \leq \tilde{\bm{z}}$.
\end{lemma}
In order to show non-resilience we first observe that condition 2. implies that in the direction $\bf{v}$ the functions $f_1 (C_1,C_2;\cdot )$ and $f_2 (C_1,C_2;\cdot)$ are both becoming positive.  By continuity of the derivative this in particular implies that $\min \{ f_1 (C_1,C_2;t {\bf{v}} ), f_2 (C_1,C_2;t {\bf{v}}) \}>0$ for all $t \in (0,t_0]$ for some $t_0>0$. We can then use Corollary~\ref{monotonicity:f} to show that also $\min \{ f_1 (\tilde{C}_1,\tilde{C}_2;t {\bf{v}} ), f_2 (\tilde{C}_1,\tilde{C}_2;t {\bf{v}}) \}>0$ for all $t \in (0,t_0]$. Then one can use the properties of $f$ derived in Proposition~\ref{properties:f} to conclude that the smallest fixed-point of $f_1 (\tilde{C}_1,\tilde{C}_2; \cdot )$ and $f_2 (\tilde{C}_1,\tilde{C}_2;\cdot)$ is bounded by $t_0 \bf{v}$. Because the choice of $t_0$ only depends on the specification of $(C_1,C_2)$ and not on the precise specification of $(\tilde{C}_1,\tilde{C}_2)$, this allows us to conclude non-resilience.

In the following we provide some examples that demonstrate the effect that the holding support can have on the resilience/non-resilience of the network. 

\begin{example}\label{expl:res1}Consider a financial network with $\P ( ( C_1,C_2)=(1,1))=1$ and $p_i>1$ for $i\in \{ 1, 2 \}$. This represents a financial system with low overall capitalization. We assume that the holdings are of type $A$.
We first consider the decoupled system $x=0$. For this we obtain 
\begin{eqnarray}
D_{\bm{v}} f_1 (\bm{z})  &=&  - v_1   + v_1 p_1    \sum_{k\leq 1} a_{1,k}(\bm{z}) = - v_1   + v_1 p_1    \sum_{k\geq 0} \psi_{0,k} (p_1 z_1,p_2 z_2) = - v_1   + v_1 p_1 e^{-z_1 p_1} \nonumber \\
D_{\bm{v}} f_2 (\bm{z})  &=&  - v_2   + v_2 p_2    \sum_{k\leq 1} a_{k,1}(\bm{z})= - v_2   + v_2 p_2    \sum_{k\geq 0} \psi_{k,0} (p_1 z_1,p_2 z_2) =- v_2   + v_2 p_2 e^{-z_2 p_2} \nonumber
\end{eqnarray}
and therefore $D_{\bm{v}} f_l (\bm{0})= -v_l + v_l p_l$ for $l\in \{1,2 \}$, which is positive for any $(v_1,v_2) > \bm{0}$ and the system is therefore non-resilient. This is as expected, because the systems are decoupled ($x=0$) and the propagation of distress in one system is not influenced by the other system. It is well known (see for example \cite{Hofstad2014}) that each system $l\in \{1,2 \}$ separately forms a giant component for $p_l >1$. The capital of each subsidiary equals $1$ ($C_l=1$) and the subsidiary defaults as soon as one of their debtors default. This implies that the default of one subsidiary in the giant component triggers the default of all subsidiaries in the giant component. 

We now look at the coupled case $x\leq -1$ and note that because of the capital structure of the uninfected network, the area of the default region that can actually be reached does not depend on $x$ for $x\leq -1$. It is therefore sufficient to consider $x=-1$. We observe that 
$$
D_{\bm{v}} f_l (\bm{z})  =  - v_l   + v_1 p_1   ( a_{0,1}(\bm{z})  + a_{1,0}(\bm{z}) )+ v_2 p_2 ( a_{0,1}(\bm{z})  + a_{1,0}(\bm{z}) )
$$
with 
\begin{eqnarray}
a_{0,1}(z_1,z_2) &=& \psi_{1,0} (p_1 z_1,p_2 z_2) =e^{-p_1 z_1 }\frac{(p_1 z_1)^1}{1!} e^{-p_2 z_2 },\\
a_{1,0}(z_1,z_2) &=& \psi_{0,1} (p_1 z_1,p_2 z_2) = e^{-p_1 z_1 } e^{-p_2 z_2 }\frac{(p_2 z_2)^1}{1!},
\end{eqnarray}
and thus $a_{0,1}(\bm{0}) = a_{1,0}(\bm{0}) =0$. It follows that $D_{\bm{v}} f_l (\bm{0}) <(0,0)$ for any $\bm{v}=(v_1,v_2)> (0,0)$. The system is therefore resilient. The fact that the system is resilient for $x=-1$ might be surprising at first for the following reason. If a subsidiary defaults, then the default of the holding's other subsidiary is triggered immediately as well. This is because when the default happens the holdings capital is $0$. From this perspective the coupling can lead to many additional defaults as a result of a defaulted subsidiary. On the other hand however, it allows to contain second order effects as healthy holdings with both subsidiaries having capital equal to $1$ require a total of 2 links to defaulted subsidiaries in order to trigger the default of one and then automatically both subsidiaries. This example shows that it can be beneficial from a regulator perspective if entities form support groups. Although the distress of one member of the group can then pull down all the others in that group, the formation of support groups leads to a more resilient system if the shock is small. 
\end{example}

\begin{example}\label{ex:resil}
We consider now a network with $\P ((C_1,C_2)=(1,0))=1$, which means that for all banks their subsidiary $1$ would be in default if it was not supported by the holding through the other subsidiary. We assume that the holdings are of type $A$. This means that the financial network is clearly more prone to contagion than the one considered in Example~\ref{expl:res1}. Again, we only need to consider the case $x=-1$. Then
$$
D_{\bm{v}} f_l (\bm{z})  =  - v_l   + v_1 p_1   a_{0,1}(\bm{z}) + v_2 p_2 a_{0,1}(\bm{z})
$$
with 
\begin{eqnarray}
a_{0,1}(z_1,z_2) &=& \psi_{0,0} (p_1 z_1,p_2 z_2) =e^{-p_1 z_1 } e^{-p_2 z_2 }.\nonumber
\end{eqnarray}
Then, we obtain that
\begin{eqnarray}
D_{(1/2,1/2)} f_l (\bm{0}) &=& 1/2( p_1 -1)  + 1/2 p_2 \nonumber
\end{eqnarray}
which is $(>0)$ for $p_1 + p_2 >1$ and $(<0)$ for $p_1 + p_2 <1$.
It follows that the system is resilient for $p_1 + p_2 <1$ and non-resilient for $p_1 + p_2 >1$. Note that in this example it is clear from the discussion about the giant component that for $p_1,p_2>1$ the system must be non-resilient because one link from a defaulted subsidiary leads automatically to the default of both subsidiaries of the same holding. Because for $p_1,p_2>1$ each subsidiary network has its own giant component, it is clear that this network is more prone to large cascades than two separated networks with the same connection probability. For $2 > p_1 + p_2 >1$ and $p_1<1$ the coupling has than actually a negative effect for the type $1$ subsidiaries network. The effect of connecting to the under-capitalized type $2$ subsidiaries network is the same as adding additional second edges (loans) between the type $1$ subsidiaries with probability $p_2$. Subsidiary $1$ inherits the exposures of subsidiary $2$. Due to these additional edges, the network becomes non-resilient. It is not surprising that the additional edges have the same effect as increasing the connection probability, because for two nodes $i$ and $j$, the probability of having neither an edge between the type $1$ subsidiaries nor between the type $2$ subsidiaries is given by $1-(p_1+p_2)/n + (p_1 p_2)/n^2$ where the term $(p_1 p_2)/n^2$ becomes negligible for large $n$. 
\end{example}

\section{Case Study}\label{casestudy}
We now study the contagion process for a number of different system configurations and illustrate a few intuitive findings. This complements the examples given in the previous section which could be studied fully analytically. For different system configurations we determine the optimal support level. We wish to highlight that the optimal support level is viewed from the point of a (benevolent) regulator, who wishes to stabilize the financial system (by some measure), by imposing rules on the financial entities to support their own subsidiaries. These support actions may not be optimal from the point of view of an individual agent, and may not be a Nash equilibrium.

We assume that all subsidiaries of type $1$ have capital $r_1$ except for an $\varepsilon_1$ share which has capital $x$ and thus defaults. Similarly, all type $2$ subsidiaries have capital equal to $r_2$ except for an $\varepsilon_2$ share that has capital $x$ and thus defaults. Assuming that the defaulted subsidiaries of type $1$ are chosen independently of those of type $2$, this of course implies that $(C_1,C_2)$ is such that $\mathbb{P}(C_l=x)=\varepsilon_l=1-\mathbb{P}(C_l=r_l)$ for $l\in \{ 1,2 \}$ and $C_1$ and $C_2$ independent.   
Then we obtain that
\begin{eqnarray}
f^A_1 (z_1,z_2)&=&  \varepsilon_1 -z_1  + (1-\varepsilon_1) \Psi_{r_1-x} (p_1 z_1) \nonumber \\
&+&(1-\varepsilon_1) \left( (1-\varepsilon_2)  \sum_{j=x+1}^{-x-1}  \psi_{r_1 + j} (p_1 z_1 ) \Psi_{\max \{r_2- j,0\}} (p_2 z_2 )  \right. \nonumber \\
&+& \left. \varepsilon_2 \sum_{j=x}^{-x-1}  \psi_{r_1 + j} (p_1 z_1 )  \right), \nonumber 
\end{eqnarray}
and
\begin{eqnarray}
f^B_1 (z_1,z_2)&=&  \varepsilon_1 -z_1  + (1-\varepsilon_1) \Psi_{r_1-x} (p_1 z_1) \nonumber \\
&+&(1-\varepsilon_1) \left( (1-\varepsilon_2)  \sum_{j=1}^{-x-1}  \psi_{r_1 + j} (p_1 z_1 ) \Psi_{\max \{r_2- j,0\}} (p_2 z_2 )  \right. \nonumber \\
&+& \left. \varepsilon_2 \sum_{j=0}^{-x-1}  \psi_{r_1 + j} (p_1 z_1 )  \right). \nonumber 
\end{eqnarray}
The functions $f^A_2$ and $f^B_2$  are simply derived by exchanging the role of $p_1, z_1, r_1$ and $p_2, z_2, r_2$. Observe that for holdings of type $A$ and for $-x\geq r_1$ respectively $-x\geq r_2$ our specification implies that the default of one subsidiary triggers immediately the default of the other subsidiary of the same holding.

In our first example we choose the parameters $r_1=5$ and $r_2=9$ and $p_1=p_2=14$, so that both layers in the networks are equally connected. We assume that in both subnetworks 15\% of the subsidiaries are defaulted, i.e. $\varepsilon_1=\varepsilon_2=0.15$. We consider a holding type $A$. The only difference between the networks of subsidiaries of types $1$ and $2$ is the capitalization. The network of subsidiaries of type $2$ is better capitalized. Also note that for $x\leq -5$, each default of subsidiary $2$ immediately triggers the default of subsidiary $1$ of the same holding, which implies that in this case the actual fraction of initially defaulted type $1$ subsidiaries is equal to $0.15+0.85 \cdot 0.15=0.2775$. The same holds for the type $2$ subsidiaries if $-x\leq 9$.
\begin{figure}[ht!]
\begin{center}
\includegraphics[scale=0.6]{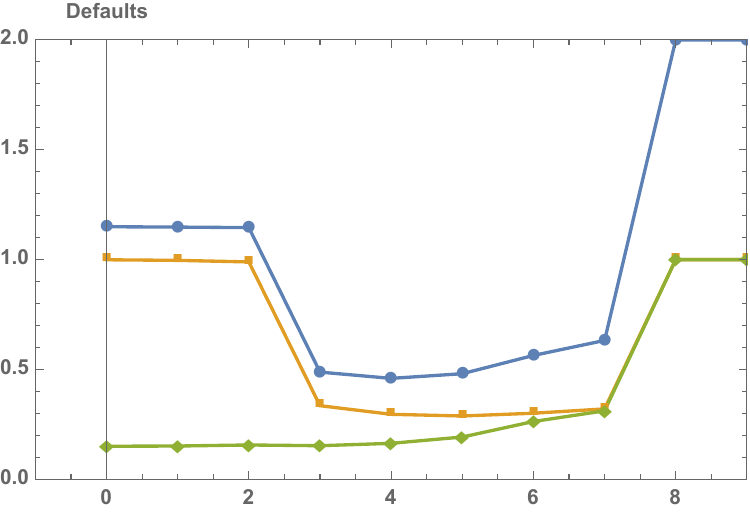}
\hspace{20pt}
\end{center}
\caption{Graph with defaults of: subsidiaries of type $1$ (orange squares),  subsidiaries of type $2$ (green rhombuses),  total defaults (blue circles), as a function of support level $x$.
The network parameters are: Holding of type $A$, $r_1=5,~r_2=9,~p_1=14,~p_2=14,~\epsilon_1=\epsilon_2=15\%$.
}
\label{fig:1} 
\end{figure}
For $-x$ ranging from $0$ to $9$ we determine the first joint zero $(\hat{z}_1, \hat{z}_2)$ of the functions $f_1^A$ and $f_2^A$ which gives us the final number of defaulted type $1$ and type $2$ subsidiaries according to Theorem~\ref{thm:threshold:model}. The results are displayed in Figure~\ref{fig:1}. This is a typical example where the number of defaults in both subsidiaries first decreases, reaches a minimum and then increases as $(-x)$ increases. While at first, this may seem counter intuitive, as more support, should mean less defaults, recall that if a holding supports its distressed subsidiary (in case of type $A$), then the holdings other subsidiary takes a hit of up to $-x$. With this in mind it is not surprising that large $-x$ may cause more type $2$ subsidiaries to default.

For $x=0$ the systems are decoupled. We observe that in this case, as a result of strong contagion, almost all type $1$ subsidiaries default. This is due to the relatively low capital in this network and the relatively strong connectivity. The contagion is very much driven by higher order effects. This can be seen by calculating the probability that a subsidiary that has not initially defaulted, defaults as a result of edges to initially defaulted subsidiaries. This probability is equal to the probability that a Poisson distributed random variable with parameter $14\cdot 0.15$ is at least $5$, and calculates as $\approx 6.21\%$ which thus accounts only for a very low share of the final defaults. Subsidiary network $2$ is different. Due to the larger capital $r_2=9$ almost no contagion takes place. This is not surprising because the probability that a subsidiary, which has not defaulted initially, defaults as the result of edges from initially defaulted subsidiaries, is very low. It is equal to  the probability that a Poisson distributed random variable with parameter $14\cdot 0.15$ is larger than $9$. This probability is less than $0.04\%$.

With support level $x<0$, the contagion in the network of subsidiaries of type $1$ can be significantly contained. In fact, we see that with a support level of $x=-3$, the fraction of defaulted type $1$ subsidiaries is reduced to $\approx 33.5 \%$ without significantly increasing the number of defaulted type $2$ subsidiaries. As a result also the total number of defaulted subsidiaries (type $1$ together with type $2$) is reduced. The lowest fraction of all the defaults is obtained for $x=-5$. With greater support, however, the subsidiaries of the type $2$ get weakened substantially up to the point where more of them default. These defaults then channel back into the network of subsidiaries of type $1$ and so forth. For $x=-8$ this leads to almost all subsidiaries defaulting.  

In comparison we now consider a network in which we change the connection probability $p_2$ to $4$ while keeping all other parameters unchanged. The results are shown in Figure~\ref{fig:2}. Due to the lower overall connectivity in the  network of subsidiaries of type $2$, the weakening of the type $2$ subsidiaries, as a result of their support for the type $1$ counterparts, has much less effect, because defaults do not amplify. In particular we do not see the feedback effects in between the two layers that we saw in the previous example. While for strong support ($x=-8$ and $x=-9$), the fraction of defaulted type $2$ subsidiaries has increased from $0.150$ to $0.209$ and $0.278$, this increase does not trigger new default rounds in the network of subsidiaries of type $1$. Observe that for both examples for $x=-9$, the number of defaulted type $1$ and type $2$ subsidiaries equal. In fact, because a subsidiary supports the holdings other subsidiary up to the point where it defaults itself, for $x=-9$, the subsidiaries always default together. 
\begin{figure}[ht!]
\begin{center}
\includegraphics[scale=0.6]{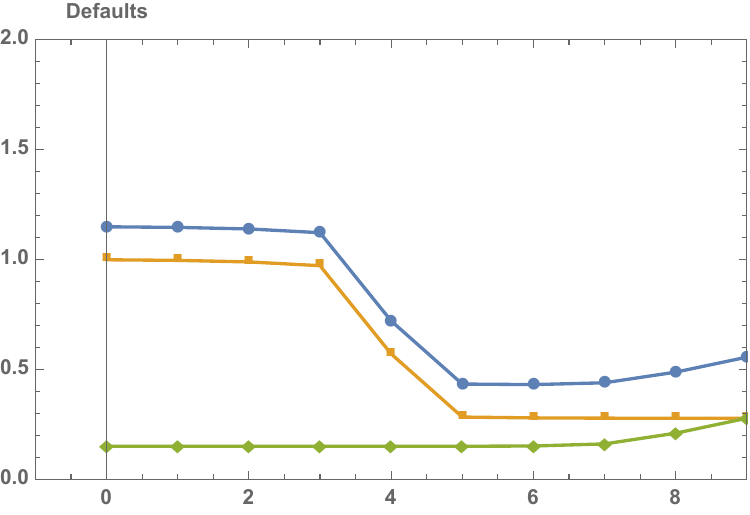}
\hspace{20pt}
\end{center}
\caption{Graph with defaults of: subsidiaries of type $1$ (orange squares),  subsidiaries of type $2$ (green rhombuses),  total defaults (blue circles), as a function of support level $x$.
The network parameters are: Holding of type $A$, $r_1=5,~r_2=9,~p_1=14,~p_2=4,~\epsilon_1=\epsilon_2=15\%$.
}
\label{fig:2} 
\end{figure}

Clearly a decrease in the initial shock $\epsilon_1$ decreases the number of subsidiaries that default (Figure \ref{fig4}). It also should be noted that at a low level of support, and with a low degree of connectivity there is very little influence on the defaults in the more capitalized subnetwork as the holding support increases.

\begin{figure}[ht!]
\begin{center}
\includegraphics[scale=0.9]{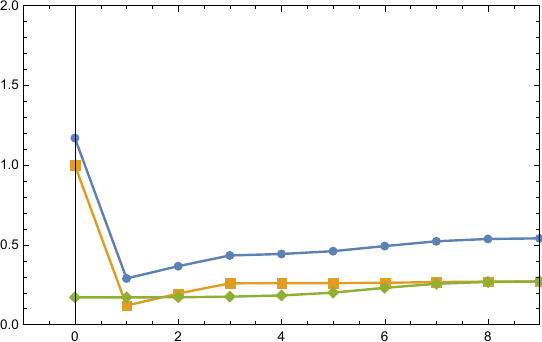}
\hspace{20pt}
\end{center}
\caption{Graph with defaults of: subsidiaries of type $1$ (orange squares),  subsidiaries of type $2$ (green rhombuses),  total defaults (blue circles), as a function of support level $x$.
The network parameters are: Holding of type $A$, $r_1=5,~r_2=9,~p_1=14,~p_2=14,~\epsilon_1=10\%,~\epsilon_2=15\%$.
}
\label{fig4} 
\end{figure}

We also illustrate that this holding structure can be a stabilizing force which ensures resilience in the sense of Section~\ref{resillience}. We consider the case from Example \ref{ex:resil}, with a holding with two (identical) subsidiaries $r_1=1=r_2,~p_1=1.5=p_2,$  and with a small initial shock: $\epsilon_1=1\%=\epsilon_2.$ Without the holding with the two networks of subsidiaries separated, there is a cluster of defaults (approximately $0.5936$), as a result of the propagation of the initial (small) shock. Whereas, once the holding is added, and the support is non-zero ($x=-1$), the defaults decrease to $0.01051$.

We also investigate the optimal distribution of wealth between the two subsidiaries assuming the wealth of the holding is fixed. Using the same parameters as before, namely $p_1=14,~p_2=4$, and 
$\varepsilon_1=\varepsilon_2=0.15$. Assume also that the wealth of each holding is $14$. The results are illustrated in Figure \ref{fig5}. Not surprisingly, we see that the minimum of total defaults is achieved when the distribution of wealth is close to equal, specifically, $r_1=8, ~r_2=6.$ Intuitively, the reason for this inequality is that subsidiaries of type $1$ are more connected, so should be capitalized a little more than subsidiaries of type $2$ so as to minimize the number of defaults in both networks of subsidiaries. 

\begin{figure}[ht!]
\begin{center}
\includegraphics[scale=0.6]{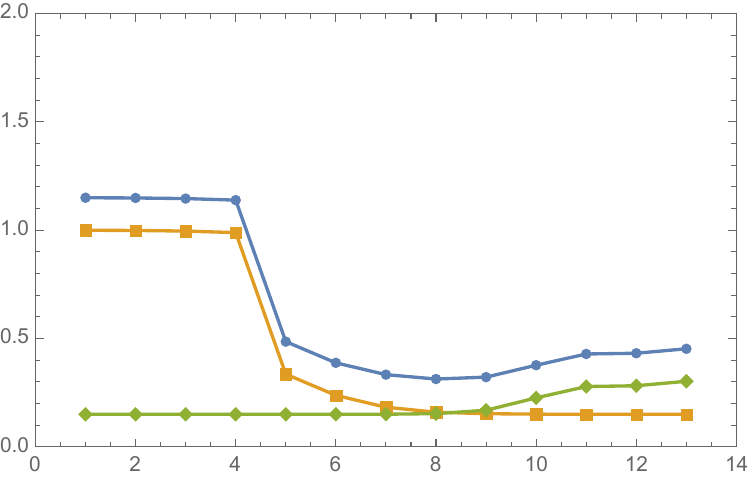}
\hspace{20pt}
\end{center}
\caption{Graph with defaults of: subsidiaries of type $1$ (orange squares),  subsidiaries of type $2$ (green rhombuses),  total defaults (blue circles) as a function of the capital $r_1$ of the subsidiaries of type $1$.
The network parameters are: Holding of type $A$, $r_2=14-r1,~x=-3,~p_1=14,~p_2=4,~\epsilon_1=\epsilon_2=15\%$.
}
\label{fig5} 
\end{figure}
As mentioned above, our setting is very homogeneous. It turns out that (at least in some cases), if banks deviate from the homogeneous support level $x$, this makes the system worse for everybody. Consider the case when $C_1=5,~C_2=9,~\epsilon_1=10\%,~\epsilon_2=15\%, ~p_1=p_2=14.$ The results are shown in Figure \ref{fig6.l}. In this case, the optimal support level is $x=-1$, as it minimizes  the number of defaulted subsidiaries of both types ($11.61\%$ of subsidiaries of type $1$ and  $15.02\%$ of type $2$), and therefore the total number of defaulted subsidiaries. However, even when only a small percentage of holdings deviates from this convention, this results in a very different outcome. We assume that $1\%$ of the holdings chose to not support their subsidiaries at all, i.e. they set their support level to zero. The remaining $99\%$ of the holdings cooperate, and support their subsidiaries up to a level of $x$. Figure \ref{fig6.r} illustrates the results. Not surprisingly for $x=-1$, the number of defaults increases compared to the situation where all holdings support up to $x=-1$. The optimal level of support for the $99\%$ of holdings increases now to $x=-2$. However, even at that level, $15.59\% $ subsidiaries of type $1$ and $15.06\%$ subsidiaries of type $2$ still default. This is a strictly worse outcome, when comparing to the original case when all the holdings support their subsidiaries to the level of $x=-1.$

\begin{figure}
\centering
\includegraphics[scale=0.9]{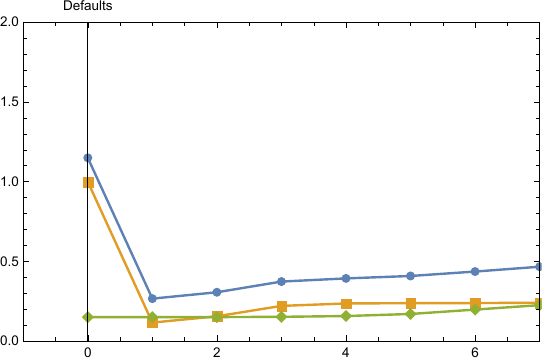}
\caption{Graph with defaults of: subsidiaries of type $1$ (orange squares),  subsidiaries of type $2$ (green rhombuses),  total defaults (blue circles), as a function of support level $x$.
The network parameters are: Holding of type $A$, $r_1=5,~r_2=9,~p_1=14,~p_2=14,~\epsilon_1=10\%,~\epsilon_2=15\%$, and all holdings support their subsidiaries at the same level of support $x$.
}
\label{fig6.l}
\end{figure}
\begin{figure}
\centering
\includegraphics[scale=0.9]{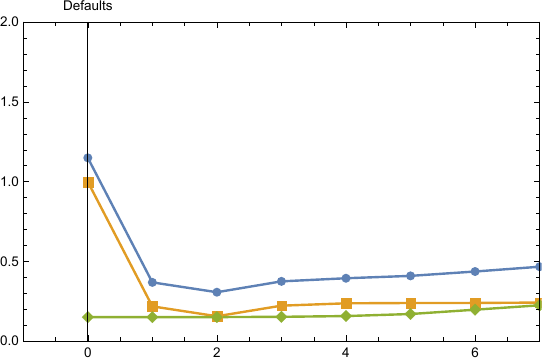}
\caption{Graph with defaults of: subsidiaries of type $1$ (orange squares),  subsidiaries of type $2$ (green rhombuses),  total defaults (blue circles), as a function of support level $x$.
The network parameters are: Holding of type $A$, $r_1=5,~r_2=9,~p_1=14,~p_2=14,~\epsilon_1=10\%,~\epsilon_2=15\%$, and $99\%$ of holdings support their subsidiaries at the same level of support $x$, whereas the other $1\%$ of holdings offer no support at all.
}\label{fig6.r}
\end{figure}

As a last example we consider again our original setting with $r_1=5$, $r_2=9$, $p_1=p_2=14$ and $\epsilon_1=\epsilon_2=15\%$, but now with a holding structure $B$. For this holding structure the support for one subsidiary has no negative effect on the other subsidiary of the same holding. As we mentioned previously, the support can be thought of as a loan that is taken against the positive capital of the holdings solvent subsidiary. The external lender is thus faced with the risk of both subsidiaries defaulting at the end. This external lender is not part of the network. As expected, we therefore see a significant reduction of the number of defaulted subsidiaries with stronger holding support and the maximal holding support $x=-9$ leads to the fewest defaults. This is illustrated in Figure \ref{fig3}.

\begin{figure}[ht!]
\begin{center}
\includegraphics[scale=0.6]{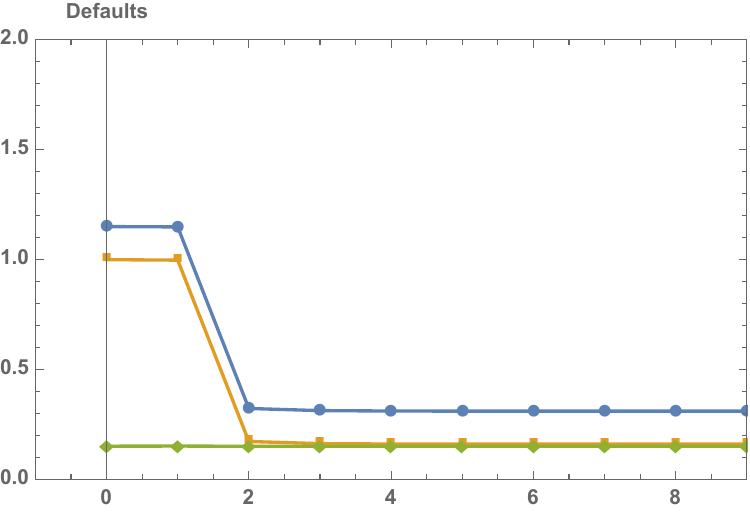}
\hspace{20pt}
\end{center}
\caption{Graph with defaults of: subsidiaries of type $1$ (orange squares),  subsidiaries of type $2$ (green rhombuses),  total defaults (blue circles), as a function of support level $x$.
The network parameters are: Holding of type $B$, $r_1=5,~r_2=9,~p_1=14,~p_2=14,~\epsilon_1=\epsilon_2=15\%$.
}
\label{fig3} 
\end{figure}

\section{Conclusion}\label{conclusion}
We proposed a simple model for a financial network of bank holdings with subsidiaries. This model allows us to study the role that the holding structure and the holding's support for troubled subsidiaries has on the propagation of contagion in a financial system. We observe that when (financial) firms are grouped into holdings this can have both, positive and negative effects. The support of a holding for a subsidiary in distress can mitigate or defer defaults and thus increase financial stability. However, in low capitalized systems, it can also amplify contagion as distress can reach previously healthy parts of the system. We show this phenomena in two ways: with a numerical case study that uses our theoretical results about the final default fraction, and analytically, based on a notion of resilience which allows us to investigate the spread of small initial shocks. We also determine the optimal support level from a regulator perspective.
We stress that this is a simple model, which allows for many interesting generalizations as discussed at the end of Section~\ref{model}, including more subsidiaries, more heterogeneous degrees, or holdings that have additional  capital.

\section{Proofs}\label{sec:proofs}
All the proofs in this section are generic for $J \in \{A, B\}$ therefore we continue to omit $J$ from the notation.
\begin{proof}[Proof of Proposition~\ref{properties:f}]
Continuity follows by continuity of the functions $\psi_k (\lambda)$ and $\Psi_k (\lambda)$ and dominated convergence. 

For part 1. recall that $f_{l}(0,0)= \sum_{(j,k) \in \mathbf{D}_{J,l}} a_{j,k}(0,0)$ with $a_{j,k}(0,0)\geq 0$. Since $\psi_{i,j}(0,0)=0$ if $i$ or $j$ is larger than $0$, it follows that $a_{j,k}(0,0)=\E [  \1_{\{ (C_1,C_2) = (j,k) \}}]=\P ((C_1,C_2) = (j,k))$ and thus $f_{l}(0,0) = \sum_{(j,k) \in \mathbf{D}_{J,l}} \P ((C_1,C_2) = (j,k))=\P((C_1,C_2)\in \mathbf{D}_{J,l})$.

Part 2. follows from \eqref{derivative}, and the fact that $a_{j,k}\geq 0$.
\end{proof}
\begin{proof}[Proof of Theorem~\ref{thm:threshold:model}]
First observe that by Proposition~\ref{properties:f} the function $f_1$ is increasing in its second argument while $f_2$ is increasing in its first argument. Moreover, $f_1(z_1,z_2)\leq 0$ for $z_1\geq 1$ and $f_2(z_1,z_2)\leq 0$ for $z_2\geq 1$ since $\sum_{(j,k) \in \mathbf{D}_{J,l}} a_{j,k}(z_1,z_2) \leq 1$. These properties allow to apply the proof of Lemma 3.2. in \cite{detering2020financial} without any changes and we can conclude that a component-wise smallest joint root $\hat{z}=(\hat{z}_1,\hat{z}_2)$ exists. 

As mentioned above, instead of directly considering the contagion process based on rounds as described in (\ref{process:generations}), we explore the set $\mathcal{S}_{i,2n-1}$ by sequentially quantifying the effect of a defaulted subsidiary on the rest of the system. So, in each step $t\in \mathbb{N}$ we only expose the effect of one (yet unexposed) defaulted subsidiary. Such a sequential reformulation is standard in the literature and was first used for exploration of the giant component (see \cite{Hofstad2014} for instance) and later used for bootstrap percolation and financial contagion (see \cite{janson2012,Cont2016,detering2019managing}). The situation here, however, is conceptually more complicated because of the complexity of the contagion process and we therefore describe it in some detail. First note that whenever a subsidiary reaches the capital $x$ we know for sure that it has defaulted, independent of the second subsidiary's financial situation. Therefore, there is no need to reduce the capital of such a subsidiary further in case it receives additional links from defaulted subsidiaries. This implies that each holding only takes on a finite number of {\em relevant} states. 

During the exploration process we keep track of several quantities:
\begin{itemize}
\item The sets $\tilde{C}_{j,k}(t)$ for $(j,k) \in \mathbf{D}$ consisting of those holdings $i\in [n]$ with capital $c_{i,1} $ of subsidiary $1$ being equal to $j$ and capital $c_{i,2}$ of subsidiary $2$ being equal to $k$ at step $t\in\mathbb{N}$. At the beginning of the process ($t=0$) we set $\tilde{C}_{j,k}(0):=\{ i \in [n] | (c_{i,1},c_{i,2})=(j,k) \}$ for all $j,k\in \mathbf{D}$. 
\item The sets $U_l(t), l\in \{1,2 \}$ of unexplored, defaulted type $l$ subsidiaries. At the beginning of the process we have $U_l(0)\subset [n_l]$ defined by: $i_l \in U_l(0)$ if and only if $(c_{i,1},c_{i,2})\in \mathbf{D}_{J,l}$. So, $U_l(0)$ contains those subsidiaries of type $l\in \{1,2\}$ which are in default at the beginning of the process.
\item The number $n_1(t)$ and $n_2 (t)$ of already explored type $1$ and type $2$ subsidiaries. We set $n_1(0)=n_2(0)=0$.
\end{itemize}

We now describe one step $t\in \mathbb{N}$ of this exploration process: At the very beginning of the step we choose uniformly a defaulted subsidiary from $U_1(t) \cup U_2(t)$. We denote this subsidiary by $S$. We then draw the links from $S$ to all the holdings in the sets $\tilde{C}_{j,k}(t)$. We then update these sets as follows. If in step $t$ we selected subsidiary $l\in \{1,2 \}$ of holding $i \in \tilde{C}_{j,k}(t)$ and it receives a link from $S$, then $i$ will be moved into $\tilde{C}_{j-1,k}(t+1)$ unless $j = x$ in which case it will just remain in its current set ($i \in \tilde{C}_{j,k}(t+1)$).
During this exploration, some subsidiaries might move into their default region. If this happens and subsidiary $1$ or $2$ defaults, we generate a {\em copy} of it and place it into the set $U_1(t)$, respectively $U_2(t)$. At the end of the step, if the selected subsidiary was in $U_1(t)$ we set $n_1(t) = n_1(t-1)+1$ and $n_2(t) = n_2(t-1)$, and we remove the subsidiary from the set of defaulted unexplored subsidiaries (i.e. we set $U_1(t+1)=U_1(t)\setminus \{ S \}$. If the selected subsidiary was from $U_2(t)$, then we set $n_2(t) = n_2(t-1)+1$ and $n_1(t) = n_1(t-1)$ and remove the subsidiary ($U_2(t+1)=U_2(t)\setminus \{ S \} $). Note that elements in $U_1(t)$ and $U_2(t)$ are distinguishable. The process ends at $\hat{t}=\min \{t \in \mathbb{N} \vert U_1(t) \cup U_2(t)=\emptyset \}$ and $\hat{t}$ corresponds to the total number of defaulted subsidiaries. Moreover, $\hat{t}= n_1(\hat{t}) + n_2 (\hat{t})$, because in each step we have exactly explored one subsidiary. Note that the procedure is such that the capital of some subsidiaries is still being updated although they have already defaulted (and possibly have been explored) but this updating has no effect. This happens if a subsidiary defaults with more capital than $x$.

We start with proving the first statement of the theorem, the lower bound for $n^{-1}\mathcal{S}^J_{i,2n-1}$ based on this sequential exploration.
Let $u_1(t)=\abs{U_1(t)}$ and $u_2(t)=\abs{U_2(t)}$ and $\tilde{c}_{j,k}(t)=\abs{\tilde{C}_{j,k}(t)}$ and let $h(t)$ describe the entire system, i.e. $$h(t)=(u_1(t), u_2(t),n_1 (t), n_2(t), \{\tilde{c}_{j,k} \}_{(j,k) \in \mathbf{D}} ).$$ 
We would like to stress that the the vector $h(t)=(u_1(t), u_2(t),n_1 (t), n_2(t), \{\tilde{c}_{j,k} \}_{(j,k) \in \mathbf{D}^x} )$ is a random vector and that we have chosen the lower case letters in order to distinguish from the random sets $U_1(t), U_2(t)$, and the $\tilde{C}_{j,k}(t)$ for $(j,k) \in \mathbf{D}$.

With probability $u_1 (t)/(u_1(t) + u_2 (t))$ in step $t$ subsidiary $1$ is selected that leads to a reduction of the number of vertices in $U_1(t)$ by $1$. On the other hand, each subsidiary in $\d \mathbf{D}_{J^{1 1}}$ receives an edge from the selected subsidiary with probability $\frac{p_1}{n}$ and thus defaults and is added to $U_1(t+1)$. Similar, if subsidiary $2$ is selected, which happens with probability $u_2 (t)/(u_1(t) + u_2 (t))$, then each subsidiary in $\d \mathbf{D}_{J^{12}}$ receives an edge from this subsidiary with probability $\frac{p_2}{n}$, defaults and is added to $U_1(t+1)$. These considerations lead to 
\begin{eqnarray}
&& \E [ u_1(t+1) -u_1 (t) \vert h(t) ] \label{eq:u-J-l}\\
&=& \frac{u_1 (t) }{u_1(t) + u_2 (t)} \left( -1 + \sum_{(j,k) \in \d \mathbf{D}_{J^{11}}}  \frac{p_1}{n} \tilde{c}_{j,k}(t) \right)
 + \frac{u_{2} (t) }{u_1(t) + u_2 (t)} \left( \sum_{(j,k) \in \d \mathbf{D}_{J^{12}} } \frac{p_{2}}{n} \tilde{c}_{j,k}(t) \right) \\
 &=& g_1 ( u_1(t)/n,u_2(t)/n,  \{ \tilde{c}_{j,k}(t)/n   \}_{(j,k) \in \mathbf{D}}  )
\end{eqnarray}
for $t\leq \hat{t} $ with 
\begin{align}
g_1(\nu_1,\nu_2, \{ \gamma_{l,m}   \}_{m,n} ):= \frac{\nu_1 }{\nu_1 + \nu_2} \left( -1 + \sum_{(j,k) \in \d \mathbf{D}_{J^{11}}}  p_1  \gamma_{j,k} \right) + \frac{\nu_2 }{\nu_1 + \nu_2}     \left( \sum_{(j,k) \in \d \mathbf{D}_{J^{12}} } p_{2} \gamma_{j,k} \right)
\label{eq:f1}
\end{align}
and similarly for $u_2(t)$. 
Moreover, a similar combinatorial argument yields that 
\begin{align}
&\E [ \tilde{c}_{k,j}(t+1) -\tilde{c}_{k,j}(t) \vert h(t) ] \\
&\qquad= \frac{u_1 (t) }{u_1(t) + u_2 (t)} \frac{p_1}{n} \left( \tilde{c}_{k+1, j} (t)\1_{\{ k+1\leq R\}}   - \tilde{c}_{k, j} (t) \1_{\{ k-1\geq x\}}  \right) \nonumber \\
&\qquad\qquad+ \frac{u_2 (t) }{u_1(t) + u_2 (t)} \frac{p_2}{n} \left( \tilde{c}_{k, j+1} (t)\1_{\{ j+1\leq R\}}  - \tilde{c}_{k, j} (t) \1_{\{ j-1\geq x\}}  \right) \nonumber \\
&\qquad= \frac{u_1 (t)/n }{u_1(t)/n + u_2 (t)/n} p_1 \left( (\tilde{c}_{k+1, j} (t)/n)\1_{\{ k+1\leq R\}}   - (\tilde{c}_{k, j} (t)/n)  \1_{\{ k-1\geq x\}}  \right) \nonumber \\
&\qquad\qquad+ \frac{u_2 (t)/n }{u_1(t)/n + u_2 (t)/n} p_2  \left( (\tilde{c}_{k, j+1} (t)/n)\1_{\{ j+1\leq R\}}  - (\tilde{c}_{k, j} (t)/n) \1_{\{ j-1\geq x\}}  \right) \nonumber \\
&\qquad= g_{k,j} (u_1 (t)/n, u_2 (t)/n, \{ \tilde{c}_{l,m}(t)/n   \}_{m,n}  )\nonumber
\end{align}
with
\begin{eqnarray}
&& g_{k,j}(\nu_1,\nu_2, \{ \gamma_{l,m}   \}_{m,n} ):= \frac{\nu_1 }{\nu_1 + \nu_2} p_1 \left( \gamma_{k+1, j}\1_{\{ k+1\leq R\}}   - \gamma_{k, j} \1_{\{ k-1\geq x\}} \right)\nonumber\\
& +& \frac{\nu_2 }{\nu_1 + \nu_2} p_2  \left( \gamma_{k, j+1}\1_{\{ j+1\leq R\}}  - \gamma_{k, j} \1_{\{ j-1\geq x\}}  \right)\nonumber
\end{eqnarray}
for $(j,k) \in \mathbf{D}$. 
And finally 
\begin{eqnarray}
\E [ n_l (t+1) -n_l (t) \vert h(t) ] &=& \frac{u_l (t) }{u_1(t) + u_2 (t)}=q_l (u_1(t)/n,u_2 (t)/n)
\end{eqnarray}
with $q_l (\nu_1,\nu_2) = \frac{\nu_l  }{\nu_1+ \nu_2 }$ for $l\in \{ 1,2 \}$.

While the above calculated differences are random and very much depend on the degree of the chosen subsidiary in each step, it turns out that the sum of a few successive steps becomes concentrated at the mean and can be described by a vector valued continuous function. For this consider therefore the vector valued function $(\nu_1 , \nu_2, \{ \gamma_{k,j} \}_{(j,k) \in \mathbf{D}},\beta_1 ,\beta_2)$ defined by the system of differential equations:
\begin{eqnarray}\label{ODE1}
\frac{\dd \nu_l (\tau)}{\dd \tau} &=& g_l(\nu_1(\tau),\nu_2 (\tau), \{ \gamma_{j,k} (\tau)  \}_{(j,k) \in \mathbf{D}}), ~ \nu_l(0) = \P ((C_1,C_2)\in \mathbf{D}_{J,l}),~l\in\{1,2\}, \label{DFG:nu:1}
\end{eqnarray}
\begin{eqnarray}\label{ODE2}
\!\!\!\!\!\!\!\!\!\frac{\dd \gamma_{k,j} (\tau)}{\dd \tau} &=& g_{k,j}(\nu_1(\tau),\nu_2 (\tau), \{ \gamma_{j,k} (\tau)  \}_{(j,k) \in \mathbf{D}}  ),~\gamma_{k,j}(0)=\P (C_1 =k,C_2=j ), ~(k,j) \in \mathbf{D}, ~~\label{DFG:gamma:1}
\end{eqnarray}
\begin{eqnarray}\label{ODE2:beta}
\frac{\dd \beta_{l} (\tau)}{\dd \tau} &=& q_l (\nu_1(\tau),\nu_2 (\tau)),~\beta_l(0)=0,  l\in\{1,2 \}.\label{DFG:beta:1}
\end{eqnarray}
Now note that the functions $g_1, g_2 , \{ g_{k,j} \}_{(j,k) \in \mathbf{D}^x}, q_1,q_2$ fulfill a Lipchitz condition on the domain
\begin{eqnarray}\label{convergence:area}
D_{\delta_0,\delta_1}
	=
	\Big\{
		(\tau ,\nu_1, \nu_2,\{ \gamma_{j,k} \}_{(j,k) \in \mathbf{D}}) \in \mathbb{R}^{5+\abs{\mathbf{D}} }
		&\mid & -\delta_0 < \tau < 2,\nonumber \\ 
		&& -\delta_0 < \gamma_{j,k} < 2,  \nonumber \\ 
		&&  -\delta_0 < \nu_1, \nu_2 <2,  \nonumber \\
		&& \delta_1 < \nu_1 + \nu_2  \Big\}\label{convergenceset},
\end{eqnarray}
for any $\delta_0,\delta_1 >0$ as can easily be seen by calculating the partial derivatives and observing that they are bounded. This implies existence of a solution to the system (\ref{ODE1}) and (\ref{ODE2}). Since the solution exists on $D_{\delta_0,\delta_1}$ for any positive $\delta_1$, we can of course extend the solution to $D_{\delta_0,0}$.

The next step of the proof is to solve the system described by (\ref{DFG:nu:1}) and (\ref{DFG:gamma:1}). In order to do so, define the following implicitly given functions
\begin{eqnarray}
h_1 (\tau ) = \frac{\nu_1(\tau)}{\nu_1 (\tau) +\nu_2 (\tau) }, \;\;\; h_2 (\tau ) = \frac{\nu_2 (\tau)}{\nu_1 (\tau) +\nu_2 (\tau) } \\
\beta_1(\tau)= \int_0^{\tau} h_1(s) ds, \;\;\; \beta_2(\tau)= \int_0^{\tau} h_2(s) ds,
\end{eqnarray}
defined on $D_{\delta_0,0}$.
Note that $h_1 (\tau ) + h_2 (\tau )=1$ and thus $\beta_1(\tau) + \beta_2(\tau) =\tau$. 
The equation for $\gamma_{R,R} (\tau)$ is a simple linear equation solved by
\begin{align}
\gamma_{R,R} (\tau) = \gamma_{R,R} (0) e^{-p_1 H_1 (\tau) } e^{-p_2 H_2 (\tau) }  =\gamma_{R,R} (0)  \psi_0 (p_1 H_1 (\tau)  ) \psi_0 (p_2 H_2 (\tau)  ) .
\end{align}
Note that $\gamma_{R,R} (\tau)$ is a constant term multiplied by the probability that two independent Poisson distributed random variables with parameters  $p_1 H_1 (\tau)$ and $p_2 H_2 (\tau)$ are zero. This allows to guess that for $(j,k) \in \mathbf{D}$,
\begin{align}
\gamma_{j,k} (\tau) =& \sum_{\substack{j\leq l_1 \leq R \\k \leq l_2\leq R}} \gamma_{l_1,l_2} (0) \left( \psi_{l_1-j} (p_1 H_1 (\tau))\1_{\{ j>x \}}  +  \Psi_{l_1-j} (p_1 H_1 (\tau))\1_{\{ j=x\}} \right) \cdot \nonumber \\
\cdot &\left( \psi_{l_2-k} (p_2 H_2 (\tau)) \1_{\{k>x \}} + \Psi_{l_2-k} (p_2 H_2 (\tau)) \1_{\{k=x \}}  \right). \label{sol:gamma}
\end{align}
To see that above expression does in fact satisfy (\ref{ODE2}), first note that
$$\frac{\partial}{\partial \lambda}\psi_k ( \lambda  )=-e^{-\lambda} \frac{\lambda^k}{k !}+ e^{-\lambda} \frac{\lambda^{k-1}}{(k-1) !}\1_{\{k\geq 1 \}}  = \psi_{k-1} ( \lambda  )\1_{\{k\geq 1 \}}-\psi_k ( \lambda  )$$
and
$$\frac{\partial}{\partial \lambda}\Psi_k ( \lambda  ) = e^{-\lambda } \frac{\lambda^{k-1} }{k-1!} =\psi_{k-1} (\lambda)\1_{\{k\geq 1 \}}.$$ 
Then for $j,k \notin \{x,R \} $ we get from (\ref{sol:gamma}) that

\begin{align}
&\frac{\partial}{\partial \tau} \gamma_{j,k} (\tau) \\
=& \sum_{\substack{j\leq l_1 \leq R \\k \leq l_2\leq R}} \gamma_{l_1,l_2} (0) p_1 h_1 (\tau) (\psi_{l_1-j-1} (p_1 H_1 (\tau))\1_{\{l_1-j\geq 1 \}} -  \psi_{l_1-j} (p_1 H_1 (\tau))) \psi_{l_2-k} (p_2 H_2 (\tau)) \nonumber \\
+& \sum_{\substack{j\leq l_1 \leq R \\k \leq l_2\leq R}} \gamma_{l_1,l_2} (0) p_2 h_2 (\tau)\psi_{l_1-j} (p_1 H_1 (\tau)) (\psi_{l_2-k-1} (p_2 H_2 (\tau))\1_{\{l_2-k\geq 1 \}} -  \psi_{l_2-k} (p_2 H_2 (\tau)))  \nonumber \\
=&p_1 h_1 (\tau) \sum_{\substack{j+1\leq l_1 \leq R \\k \leq l_2\leq R}} \gamma_{l_1,l_2} (0)  (\psi_{l_1-(j+1)} (p_1 H_1 (\tau))  \psi_{l_2-k} (p_2 H_2 (\tau)) \nonumber \\  
 -&p_1 h_1 (\tau) \sum_{\substack{j\leq l_1 \leq R \\k \leq l_2\leq R}} \gamma_{l_1,l_2} (0)   \psi_{l_1-j} (p_1 H_1 (\tau)) \psi_{l_2-k} (p_2 H_2 (\tau)) \nonumber \\
 +&p_2 h_2 (\tau) \sum_{\substack{j\leq l_1 \leq R \\k+1 \leq l_2\leq R}} \gamma_{l_1,l_2} (0)  (\psi_{l_1-j} (p_1 H_1 (\tau))  \psi_{l_2-(k+1)} (p_2 H_2 (\tau)) \nonumber \\ 
  -&p_2 h_2 (\tau) \sum_{\substack{j\leq l_1 \leq R \\k \leq l_2\leq R}} \gamma_{l_1,l_2} (0)   \psi_{l_1-j} (p_1 H_1 (\tau)) \psi_{l_2-k} (p_2 H_2 (\tau)) \nonumber \\
  =& p_1 h_1 (\tau)  \gamma_{j+1,k} (\tau) - p_1 h_1 (\tau) \gamma_{j,k} (\tau)  + p_2 h_2 (\tau) \gamma_{j,k+1} (\tau) - p_2 h_2 (\tau)  \gamma_{j,k} (\tau),
\end{align}
which verifies (\ref{ODE2}). The case where $j$ or $k$ is in $\{x,R \}$ can be verified similarly and is left to the reader. 

In order to derive $\nu_1(\tau)$ and $\nu_2(\tau)$, recall that
$$\nu_l(0) = \P ((C_1,C_2)\in \mathbf{D}_{J,l})=\sum_{(r_1,r_2)\in \mathbf{D}^x_{J,l}} \!  \! \! \! \!  \! \! \! \gamma_{r_1,r_2}(0),~l\in \{ 1,2 \}.$$
We then make the following Ansatz for the solution:
\begin{equation}\label{sol:nu}
\nu_l(\tau) =\sum_{(j,k) \in \D_{J,l}} \gamma_{j,k} ( \tau ) -H_1(\tau) .
\end{equation}
Note that the $\nu_l$ are well defined on $D_{\delta_0,0}$ as the sum is only over finitely many terms. 
A tedious, but straightforward calculation of taking derivatives of \eqref{sol:nu}, verifies that 
$\nu^J_l$ satisfies \eqref{DFG:nu:1} as most of the terms in the sum cancel out and only the terms on the boundaries $\d \mathbf{D}_{J^{ml}}, l,m \in \{1,2 \}$ remain.

For $(k,j) \in \D$ we now want to approximate the quantities $\tilde{c}_{j,k} (t)$ with $ \gamma_{j,k} (t/n)$, and for $i \in \{ 1,2\}$ the quantities $u_i (t)$ with $\nu_i (t/n)$ respectively. We use \cite{10.1214/aoap/1177004612}[Theorem 2] to show that in fact
\begin{eqnarray}
u_l (t)&=&  \nu_{l} (t/n) n + o(n), l\in \{ 1,2\} ,\label{approx:u} \\
\tilde{c}_{j,k} (t) &= &  \gamma_{j,k} (t/n) n + o(n), (k,j) \in \D ,\label{approx:c} \\
n_l (t)&=& \beta_l (t/n) n + o(n), l\in \{ 1,2\} ,\label{approx:n}
\end{eqnarray}
with high probability. 
By above considerations, assumptions (ii) and (iii) of \cite{10.1214/aoap/1177004612}[Theorem 2] are satisfied. To see that also assumption (i) is satisfied, note that in each step of the exploration, the number of vertices receiving an edge from the currently explored defaulted subsidiary is bounded by the degree of this subsidiary. Let $D^+$ be the largest out-degree in the graph. A simple probabilistic bound for $D^+$ is given by 
\[ 
\mathbb{P} (D^+ \geq \log (n) ) \leq n \binom{n-1}{\log (n)} (\max \{p_1,p_2 \} /n)^{\log (n)} \leq n (\max \{p_1,p_2 \} )^{\log(n)}/(\log (n) !).
\]
By Stirling's formula, it follows that $\log(n) ! \geq n^m (\max \{p_1, p_2 \} )^{\log(n)}$ for any $m$ and thus
\[
\mathbb{P} (D^+ \geq \log (n) ) \leq n^{-(m-1)}.
\]
We may choose $\lambda (n) =n^{1/10}$ and $\omega (n) = n^{5/10}$ for condition (i) in Wormald's Theorem. This implies that the approximation (\ref{approx:u})-(\ref{approx:n})
holds uniformly for $$(t/n ,\nu_1 (t/n), \nu_2 (t/n),\{ \gamma_{j,k} (t/n) \}_{(j,k) \in \mathbf{D}}, \beta_1(t/n), \beta_2 (t/n)) \in D_{\delta_0,\delta_1}$$ with high probability.

Let now 
$\tau_{\delta_1}:= \sup \{ \tau \in \mathbb{R}_+ : (\tau ,\nu_1, \nu_2,\{ \gamma_{j,k} \}_{(j,k) \in \mathbf{D}}) \in D_{\delta_0,\delta_1} \}$. By the definition of the system $(u_1(t), u_2(t),\{ \tilde{c}_{k,j}(t) \}_{(j,k) \in \mathbf{D}} )$ it is clear that all quantities are in between $0$ and $2n$ and thus by (\ref{approx:u})-(\ref{approx:n}) and the definition of $D_{\delta_0,\delta_1}$ it follows that $\nu_1(\tau_{\delta_1})+ \nu_2(\tau_{\delta_1})= \delta_1$.

Let us now set $\hat{\tau}:=\lim\limits_{\delta_1 \rightarrow 0} \tau_{\delta_1}$, where the limit clearly exists as $\tau_{\delta_1}$ is increasing in $\delta_1$ and bounded by $2$. Moreover, $\lim\limits_{\delta_1 \rightarrow 0}\nu_1(\tau_{\delta_1})=\lim\limits_{\delta_1 \rightarrow 0}\nu_2(\tau_{\delta_1})= 0$ and we shall thus define $\nu_1(\hat{\tau})=\nu_2(\hat{\tau})=0$.

Note that $f_l (\beta_1(\tau) , \beta_2(\tau)) = \nu_l (\tau ) $ for $\tau\leq \hat{\tau}$. Moreover, since $(z_1,z_2)$ is the smallest (component-wise) joint zero of $f_1$ and $f_2$, and due to the monotonicity and continuity of $\beta_1$ and $\beta_2$, and of $f_l$ as derived in Proposition~\ref{properties:f}, it follows that $(\beta_1(\hat{\tau}) , \beta_2(\hat{\tau}))=(\hat{z}_1, \hat{z}_2)$. Let $\hat{t}:= \min \{ t \;\vert \; u_1(t)+ u_2(t) =0 \}$. By the previous considerations it follows that $\hat{t}/n\geq \tau_{\delta_1}$ w.h.p. for every $\delta_1$ and thus by (\ref{approx:u}) $n^{-1}\mathcal{S}_{i,2n-1} \geq \beta_i(\tau_{\delta_1})- \varepsilon
$ w.h.p. for any $\varepsilon >0$ and by $\lim\limits_{\delta_1\rightarrow 0} \beta_i(\tau_{\delta_1})=\hat{z}_i$ it follows that
$$
n^{-1}\mathcal{S}_{i,2n-1} \geq \hat{z}_i- \varepsilon
$$
with high probability. This proves the first part of the theorem.

Now we will show that under the additional assumption that there exists $\bf{v}= (v_1, v_2) \in \mathbb{R}^2_+$ such that 
$D_{\bf{v}} f_1 (\hat{z}), D_{\bf{v}} f_2 (\hat{z})  <0$, it holds in fact that 
\[ n^{-1}\mathcal{S}_{l,2n-1} \xrightarrow{p}  \hat{z}_i,\quad\text{as }n\to\infty. \] 

We know that $D_{\bf{v}} f_1 (\hat{z})=D_{\bf{v}} f_1 (\beta_1(\hat{\tau}) , \beta_2(\hat{\tau})), D_{\bf{v}} f_2 (\hat{z})= D_{\bf{v}} f_2 (\beta_1(\hat{\tau}) , \beta_2(\hat{\tau})) <0$ and by continuity of the derivative and because of
$$\lim_{\tau \rightarrow \hat{\tau}} (\beta_1(\tau) , \beta_2(\tau))=(\beta_1(\hat{\tau}) , \beta_2(\hat{\tau}))= (\hat{z}_1 ,\hat{z}_2)$$
it follows that also $D_{\bf{v}} f_1 (\beta_1(\tau_{\delta_1}) , \beta_2(\tau_{\delta_1})), D_{\bf{v}} f_2 (\beta_1(\tau_{\delta_1}) , \beta_2(\tau_{\delta_1}))  <0$ for $\delta_1$ small enough.
Since $a_{j,k} (\beta_1(\tau),\beta_2(\tau))=\gamma_{j,k} (\tau)$ it follows by (\ref{deriv:f}) that
\begin{eqnarray}
 &&  v_1 p_1 \sum_{(j,k) \in  \d \mathbf{D} _{J^{11}}  }  \gamma_{j,k}(\tau_{\delta_1}) + v_2 p_2 \sum_{(j,k) \in  \d \mathbf{D} _{J^{12}}  }  \gamma_{j,k}(\tau_{\delta_1}) <c_1 v_1 \label{stop:cond:1} \\
  &&   v_1 p_1 \sum_{(j,k) \in  \d \mathbf{D} _{J^{21}}  }  \gamma_{j,k}(\tau_{\delta_0}) + v_2 p_2 \sum_{(j,k) \in  \d \mathbf{D} _{J^{22}}  }  \gamma_{j,k}(\tau_{\delta_1}) <c_2 v_2 \label{stop:cond:2} 
\end{eqnarray}
for some $c_1,c_2 <1$ sufficiently close to $1$. We shall track the process until the step $\lfloor \tau_{\delta_1} n \rfloor$ by making use of (\ref{approx:u}) - (\ref{approx:n}), and then we explore the remaining subsidiaries in the sets $U_1(\lfloor \tau_{\delta_1} n \rfloor)\cup U_2(\lfloor \tau_{\delta_1} n \rfloor)$ in rounds. We place the holdings with capital $(j,k) \in \mathbf{D}$ in four sets which we define now and where a holding might be placed in more than one set. For this, let us call a subsidiary {\em weak} if it can default by either receiving one edge from a defaulted subsidiary itself or because the other subsidiary of the same holding receives a link from a defaulted subsidiary. We call a subsidiary {\em strong} if it can only default if this subsidiary and the other subsidiary of the same holding, receive in total at least two links from defaulted subsidiaries. For $l\in \{1,2 \}$ let $\W_l$ be the set that includes those holdings with weak type $l$ subsidiaries and $\St_l$ those with strong type $l$ subsidiaries. Note that $\W_l \cap \St_l =\emptyset$ for $l\in \{ 1,2\}$ but the other four possible intersections of these sets are not necessary empty. Further for each $l\in \{1,2 \}$ we define the subsets $\W^i_l$ of holdings whose weak type $l$ subsidiaries default in round $i$ of the exploration process and the subsets $\St^i_l$ of holdings whose strong type $l$ subsidiaries default in round $i$ of the exploration process.

Note that from (\ref{stop:cond:1}) and (\ref{stop:cond:2}) we get together with the approximation (\ref{approx:c}) that for $n$ large enough

\begin{eqnarray}
 &&  v_1 p_1 \sum_{(j,k) \in  \d \mathbf{D} _{J^{11}}  }  \frac{\tilde{c}_{j,k}( \lfloor \tau_{\delta_1}n \rfloor )}{n} + v_2 p_2 \sum_{(j,k) \in  \d \mathbf{D} _{J^{12}}  }   \frac{\tilde{c}_{j,k}( \lfloor \tau_{\delta_1}n \rfloor )}{n} <c_1 v_1  \label{v1Ineq}\\
  &&   v_1 p_1 \sum_{(j,k) \in  \d \mathbf{D} _{J^{21}}  }   \frac{\tilde{c}_{j,k}( \lfloor \tau_{\delta_1}n \rfloor )}{n} + v_2 p_2 \sum_{(j,k) \in  \d \mathbf{D} _{J^{22}}  }   \frac{\tilde{c}_{j,k}( \lfloor \tau_{\delta_1}n \rfloor )}{n} <c_2 v_2 \label{v2Ineq}
\end{eqnarray}
w.h.p..
This allows us now to bound the expected number of weak subsidiaries of type $l\in \{1,2 \}$ that default in the first round through a combinatorial counting argument by
\begin{eqnarray} 
&&n^{-1} \E\left[ \abs{\W^1_l} \Big\vert h( \lfloor \tau_{\delta_1} n \rfloor )\right] \nonumber \\ 
& \leq & \left( \frac{u_1( \lfloor \tau_{\delta_1}n \rfloor )}{n} p_1 \sum_{(j,k) \in  \d \mathbf{D} _{J^{l1}}}  \frac{\tilde{c}_{k,j}(\lfloor \tau_{\delta_1}n \rfloor)}{n}  + \frac{u_2( \lfloor \tau_{\delta_1}n \rfloor )}{n} p_2 \sum_{(j,k) \in  \d \mathbf{D} _{J^{l2}}}   \frac{\tilde{c}_{k,j}(\lfloor \tau_{\delta_1}n \rfloor)}{n} \right) \nonumber \\ \nonumber
& < &   c_l v_l \varepsilon 
\end{eqnarray}
and in particular
\begin{equation}
n^{-1} \sum_{l\in \{1,2 \}}\E\left[   \abs{\W^1_l} \Big\vert h( \lfloor \tau_{\delta_1} n \rfloor )\right]  <  \sum_{l\in \{1,2 \}}  c_l v_l \varepsilon 
\end{equation}
for any $\varepsilon$ provided that $\frac{u_1( \lfloor \tau_{\delta_1}n \rfloor )}{n}< v_1\varepsilon$ and $\frac{u_2( \lfloor \tau_{\delta_1}n \rfloor )}{n}< v_2\varepsilon$, which we can ensure by possibly decreasing $\delta_1$ further. We further choose $c\in (\max \{c_1,c_2 \},1)$ and obtain that  $$\sum_{l\in \{1,2 \}}  c_l v_l \varepsilon < c \varepsilon  (v_1+v_2). $$
Similarly we find for $l\in \{1,2 \}$ the very rough bound
\begin{eqnarray} 
&&n^{-1} \E\left[ \abs{\St^1_l} \Big\vert h( \lfloor \tau_{\delta_1} n \rfloor )\right] \nonumber \\ 
& = &\left( \frac{ u_1( \lfloor \tau_{\delta_1}n \rfloor ) + u_2( \lfloor \tau_{\delta_1}n \rfloor )  }{n}\right)^2 \max \{ p_1,p_2 \} \sum_{(j,k) \in  \mathbf{D}^c_{J,l}}  \frac{\tilde{c}_{k,j}(\lfloor \tau_{\delta_1}n \rfloor)}{n}  \\ \nonumber
 &<& ( ( v_1 + v_2 )\varepsilon)^2 \max \{ p_1,p_2 \}  < (c'/4) \varepsilon v_l
\end{eqnarray}
with $c'>0$ such that $2c' < 1- c$, provided that $\varepsilon $ is chosen small enough. In particular it then follows that 
\begin{equation}
n^{-1} \sum_{l\in \{1,2 \}} \E \left[  \W^{1}_l + \St^{1}_l  \right] < (c+c'/2) \varepsilon ( v_1 + v_2 ) .
 \end{equation}
Our plan is now to show that for $l \in \{1,2 \}$ it holds for every $m>0$ that
\begin{eqnarray} 
\E \left [n^{-1} \W^{m}_l \right] < (c_l + c'/4)(c + c')^{(m-1)}  \varepsilon v_l \label{boundW} \\
 \E \left[n^{-1} \St^{m}_l \right] <  (c'/4) (c + c')^{(m-1)}  \varepsilon  v_l.  \label{boundS}
\end{eqnarray} 
This implies 
\begin{equation}\label{boundWS}
 n^{-1} \sum_{l\in \{1,2 \}}\E \left[   \W^{m}_l + \St^{m}_l   \right] < (c +c')^m  \varepsilon ( v_1 + v_2 ),
\end{equation}
and from this it follows that 
\begin{equation}
 n^{-1} \sum_{m=1}^{2n} \sum_{l\in \{1,2 \}}\E \left[  \W^{m}_l + \St^{m}_l   \right] < \frac{1}{1-(c+c')} \varepsilon ( v_1 + v_2 ).
 \end{equation}
 and therefore 
 \begin{equation}\label{MI}
 \P \left(  n^{-1} \sum_{m=1}^{2n} \sum_{l\in \{1,2 \}} ( \W^{m}_l + \St^{m}_l ) \geq \sqrt{\frac{1}{1-(c+c')} \varepsilon ( v_1 + v_2 )} \right)\leq \sqrt{\frac{1}{1-(c+c')} \varepsilon ( v_1 + v_2 )}.
  \end{equation}
The quantity on the right hand side of (\ref{MI}) can be made arbitrarily small by decreasing $\varepsilon$. So in order to conclude, it remains to show (\ref{boundW}) and (\ref{boundS}).

For $m=1$ this is done above. Suppose that (\ref{boundW}) and (\ref{boundS}) holds for $m>0$. Then it follows for $l\in \{1,2 \}$ by (\ref{v1Ineq}) and (\ref{v2Ineq}) that

\begin{eqnarray}
& & n^{-1} \E\left[ \abs{\W^{m+1}_l} \Big\vert h( \lfloor \tau_{\delta_1} n \rfloor )\right] \\
&\leq & \frac{\sum_{x \in \W_1 \cup \St_1 } \P (x \in \W_1^m \cup \St_1^m )}{n} p_1 \sum_{(j,k) \in  \d \mathbf{D} _{J^{l1}}}  \frac{\tilde{c}_{k,j}(\lfloor \tau_{\delta_1}n \rfloor)}{n} \\
& + &  \frac{\sum_{x \in \W_2 \cup \St_2}  \P (x \in \W_2^m \cup \St_2^m ) }{n} p_2 \sum_{(j,k) \in  \d \mathbf{D} _{J^{l2}}}   \frac{\tilde{c}_{k,j}(\lfloor \tau_{\delta_1}n \rfloor)}{n} \nonumber \\
& \leq &  (\E [n^{-1} \W^{m}_1 ] + \E [n^{-1} \St^{m}_1 ] )  p_1 \sum_{(j,k) \in  \d \mathbf{D} _{J^{l1}}}  \frac{\tilde{c}_{k,j}(\lfloor \tau_{\delta_1}n \rfloor)}{n} \\
& + & (\E [n^{-1} \W^{m}_2 ] + \E [n^{-1} \St^{m}_2 ] )  p_2 \sum_{(j,k) \in  \d \mathbf{D} _{J^{l2}}}   \frac{\tilde{c}_{k,j}(\lfloor \tau_{\delta_1}n \rfloor)}{n} \nonumber \\
&<& (c_l + c'/2) (c + c' )^m \varepsilon v_l
\end{eqnarray}
and in particular 
$$n^{-1} \E[ \sum_{l\in \{1,2 \}} \abs{\W^{m+1}_l} \vert h( \lfloor \tau_{\delta_1} n \rfloor )] < (c + c'/2 ) (c + c' )^m \varepsilon (v_1+v_2) $$

For a subsidiary to be in $\St^{m+1}_l$ it must receive a link from the set $\W^{m}_l \cup \St^{m}_l$ and one from $\bigcup_{k=1}^m \W^{k}_l \cup \St^{k}_l$. It follows from a very rough bound that 
\begin{eqnarray}
& & n^{-1} \E[ \abs{\St^{m+1}_l} \vert h( \lfloor \tau_{\delta_1} n \rfloor )] \nonumber \\
&\leq & \frac{\sum_{x \in \W_1 \cup \St_1 \cup \W_2 \cup \St_2 } \P (x \in \W_1^m \cup \St_1^m \cup \W_2^m \cup \St_2^m )}{n} \max \{ p_1,p_2 \} \cdot \nonumber \\
& \cdot & \frac{\sum_{x \in  \W_1 \cup \St_1 \cup \W_2 \cup \St_2 } \P (x \in \bigcup_{k=1}^{m} \W_1^m \cup \St_1^m \cup \W_2^m \cup \St_2^m )}{n} \max \{ p_1,p_2 \} \nonumber \\
&<& (c + c'/2)^2 (c + c' )^m 1/(1- c + c' ) \varepsilon^2 (v_1 + v_2)^2 \nonumber \\
& < & c'/4 (c + c' )^m \varepsilon (v_1 + v_2)
\end{eqnarray}
for $\varepsilon$ sufficiently small. This finishes the proof.
\end{proof}
\begin{proof}[Proof of Proposition~\ref{representationf:P}] 
First note that the graph that describes the bound of the default region for subsidiary $1$ is a non-decreasing function of the capital of subsidiary $2$. Similarly the graph that describes the bound of the default region of subsidiary $2$ is a non-decreasing function of the capital of subsidiary $1$ (see Figure \ref{fig1}). We rewrite (\ref{eq:f-a}): 
 \begin{eqnarray}
 f_l(C_1,C_2;z_1,z_2)&=&\sum_{(j,k) \in \mathbf{D}_{J,l}} a_{j,k}(z_1,z_2)  -z_l\nonumber \\
 &=&\sum_{(j,k) \in \mathbf{D}_{J,l}} \E [  \1_{\{ (C_1,C_2) \geq (j,k) \}} \psi_{C_1-j,C_2-k} (p_1 z_1,p_2 z_2) ] -z_l \nonumber \\
  &=&\sum_{(j,k) \in \mathbf{D}_{J,l}} \sum_{(n,h)\geq (j,k)} \E [  \1_{\{ (C_1,C_2) = (n,h) \}} \psi_{n-j,h-k} (p_1 z_1,p_2 z_2) ]-z_l\nonumber \\
  &=& \sum_{(n,h)} \sum_{\substack{(j,k) \in \mathbf{D}_{J,l}\\(j,k) \leq (n,h)}} \E [  \1_{\{ (C_1,C_2) = (n,h) \}} \psi_{n-j,h-k} (p_1 z_1,p_2 z_2) ]-z_l\label{lastline}.
  \end{eqnarray}
  For fixed $(n,h)$ because of $\sum_{k,l\geq 0}\psi_{k,l} (p_1 z_1,p_2 z_2)=1$, it follows that $$\sum_{(j,k)\leq (n,h)}\E [  \1_{\{ (C_1,C_2) = (n,h) \}} \psi_{n-j,h-k} (p_1 z_1,p_2 z_2) ]=\P((C_1,C_2) = (n,h)).$$

  The shape of the default region is such that if $(n,h) \in \mathbf{D}_{J,l}$ and $(j,k)\leq  (n,h)$, then also $(j,k) \in \mathbf{D}_{J,l} $. This implies that in the inner sum in (\ref{lastline}) for $(n,h)\in \mathbf{D}_{J,l}$ we can actually drop the condition $(j,k)\in \mathbf{D}_{J,l}$. Thus we can rewrite the double sum in (\ref{lastline}) as \begin{equation}
  \P((C_1,C_2) \in \mathbf{D}_{J,l} ) + \sum_{(n,h)\in \mathbf{D}^c_{J,l}} \sum_{\substack{(j,k) \in \mathbf{D}_{J,l}\\(j,k) \leq (n,h)}} \E [  \1_{\{ (C_1,C_2) = (n,h) \}} \psi_{n-j,h-k} (p_1 z_1,p_2 z_2) ].\label{poisson:2}
  \end{equation}
  Now observe that $\E [  \1_{\{ (C_1,C_2) = (n,h) \}} \psi_{n-j,h-k} (p_1 z_1,p_2 z_2) ]/\P((C_1,C_2) = (n,h))$ is the conditional expectation of the random variable $\psi_{C_1-j,C_2-k} (p_1 z_1,p_2 z_2)$, given that $(C_1,C_2) = (n,h)$ and $$(1/\P((C_1,C_2) = (n,h))) \cdot \sum_{\substack{(j,k) \in \mathbf{D}_{J,l}\\(j,k) \leq (n,h)}} \E [  \1_{\{ (C_1,C_2) = (n,h) \}} \psi_{n-j,h-k} (p_1 z_1,p_2 z_2) ],$$ the conditional expectation of the random variable 
  
  \begin{equation}
      \Lambda^l_{C_1,C_2} (p_1 z_1, p_2 z_2 )=\sum_{\substack{(j,k) \in \mathbf{D}_{J,l}\\(j,k) \leq (n,h)}} \psi_{C_1-j,C_2-k} (p_1 z_1,p_2 z_2) ,\label{poisson:D}
        \end{equation}
        given that $(C_1,C_2) = (n,h)$. By the definition of $\psi$, this is just the probability that $(n,h)+(X_1(p_1 z_1),X_2(p_2 z_2))$ is in the default region where $X_1$ and $X_2$ being independent Poisson distributed random variables with parameter $p_1 z_1$ and $p_2 z_2$ respectively. Summing over $(n,h)$, as we do in (\ref{poisson:2}), now gives us by the law of total expectation the unconditional probability of $\Lambda_{C_1,C_2} (p_1 z_1, p_2 z_2 )$ which shows (\ref{representation:f}).
        By the law of total probability it follows that 
   \begin{equation*}
f_l(C_1,C_2;z_1,z_2)= \P \left( (C_1,C_2)+(X_1(p_1 z_1),X_2(p_2 z_2)) \in \mathbf{D}_{J,l}\right),
   \end{equation*}
    with $X_1(p_1 z_1)$ and $X_2(p_2 z_2)$ being mutually independent Poisson distributed random variables with parameter $p_1 z_1$ and $p_2 z_2$ independent of $(C_1,C_2)$.
\end{proof}
\begin{proof}[Proof of Corollary~\ref{corollary:f:cont:in:C}]
We prove the result for $f_l$, the proof for $D_{\bm{v}} f_l$ is very similar. For fixed $(z_1,z_2)\in \mathbb{R}_+\times \mathbb{R}_+$, point-wise convergence follows directly from Proposition~\ref{representationf:P} since for every $(z_1,z_2)$, the integrand $\Lambda_{C_1,C_2} (z_1 p_1 , z_2 p_2 )$ in (\ref{representation:f}) is bounded by $1$. By weak convergence of $(C_1^n,C_2^n)$ it then follows that $f_l(C^n_1,C^n_2;z_1,z_2)$ converges to $f_l(C_1,C_2;z_1,z_2)$. Let now $A\subset \mathbb{R}_+ \times \mathbb{R}_+$ compact. For simplicity we assume that $A$ is an interval of the form $A=[a,b]$ with $a,b \in \mathbb{R}_+ \times \mathbb{R}_+$ but the proof for general compact sets does not pose any further complications. For a given $\delta$, choose now $K\in \mathbb{N}$ and points $x^1=(x_1^1,x_2^1),\dots , x^K=(x^K,x_2^K) \in A$ such that for every $x=(x_1,x_2)\in A$ there exists $0\leq k\leq K$, such that $\abs{x-x_k}\leq \delta/2$. Now choose $n$ large enough such that $\abs{f_l(C^n_1,C^n_2;x_1^k,x_2^k)-f_l(C_1,C_2;x_1^k,x_2^k)}\leq \delta/2 (1+\max \{p_1,p_2 \})$ for all $0\leq k\leq K$. By (\ref{derivative}), the partial derivatives of $f_l$ are bounded by $1+\max \{p_1,p_2 \} $. Therefore, for arbitrary $x$, choose $x_k$ such that $\abs{x-x_k}\leq \delta/2$. By the triangular inequality 
\begin{eqnarray}
&& \abs{f_l(C^n_1,C^n_2;x_1,x_2) - f_l(C_1,C_2;x_1,x_2)} \nonumber \\
& = & \abs{f_l(C^n_1,C^n_2;x_1,x_2) - f_l(C^n_1,C^n_2;x^k_1,x^k_2) + f_l(C^n_1,C^n_2;x^k_1,x^k_2) - f_l(C_1,C_2;x_1,x_2)}\nonumber \\
 &\leq & \abs{f_l(C^n_1,C^n_2;x_1,x_2) - f_l(C^n_1,C^n_2;x^k_1,x^k_2)} + \abs{f_l(C^n_1,C^n_2;x^k_1,x^k_2) - f_l(C_1,C_2;x_1,x_2)}\nonumber \\
 & \leq & \delta, \nonumber
\end{eqnarray}
which shows uniform convergence. 
\end{proof}
\begin{proof}[Proof of Corollary~\ref{monotonicity:f}]
The first statement follows directly from the representation (\ref{representation:f:P}) since by the shape of the default region for $(j,k) < (n,h)$, it holds that $$\Lambda^l_{j,k} (p_1 z_1, p_2 z_2 )\geq  \Lambda^l_{n,h} (p_1 z_1, p_2 z_2 ).$$ 
To prove 2. observe that for $(j,k) < (n,h)$, it actually holds that $$\Lambda^l_{j,k} (p_1 z_1, p_2 z_2 )> \Lambda^l_{n,h} (p_1 z_1, p_2 z_2 )$$ for $(z_1,z_2)\neq (0,0)$ unless $(j,k), (n,h)\in \mathbf{D}_{J,l}$. The condition $\P ((\tilde{C}_1,\tilde{C}_2) \leq (j,k))< \P ((C_1,C_2) \leq (j,k))$ for some $(j,k)\in \mathbf{D}^c_{J,l}\cup \mathbf{B}_{J,l}$ ensures then exactly that $\P (\Lambda_{\tilde{C_1},\tilde{C_1}} (p_1 z_1, p_2 z_2 ) \leq (j,k))< \P (\Lambda_{C_1,C_1} (p_1 z_1, p_2 z_2 ) \leq (j,k))$ and by (\ref{representation:f}) the claim follows. 
\end{proof}
\begin{proof}[Proof Theorem~\ref{resilience}]
We start with the proof for 1. Let $\tilde{\varepsilon}>0$ be given. We know by Corollary~\ref{corollary:f:cont:in:C} that if $(\tilde{C}_1,\tilde{C}_2)$ and $(C_1,C_2)$ are close (in distribution), then the functionals and their derivatives are close. Now, if there exists a direction $\bf{v} \in \mathbb{R}^2_{+}$ such that $D_{\bf{v}} f_l (C_1,C_2;\bm{0}) <0$, then for $\varepsilon$ small enough, we know that also $D_{\bf{v}} f_l (\tilde{C}_1,\tilde{C}_2;\bm{0}) =: \kappa <0$ for $\P ((\tilde{C}_1,\tilde{C}_2) \neq(C_1,C_2))\leq \varepsilon$. In fact, it follows by continuity that $D_{\bf{v}} f_l (\tilde{C}_1,\tilde{C}_2; \bm{z}) < \kappa/2 $ for $\bm{z}$ in some neighborhood $B_{\delta}(\bm{0})$ of $\bm{0}$. By possibly decreasing $\varepsilon$ even further to the point that $f_1 (\tilde{C}_1,\tilde{C}_2;\bm{0}),f_2 (\tilde{C}_1,\tilde{C}_2;\bm{0}) \leq -\tilde{\varepsilon} \kappa /4$ and ensuring that $\varepsilon/2 \leq \delta$, we get that $ f_l (\tilde{C}_1,\tilde{C}_2;(\varepsilon/2) {\bf{v}}) <0$ for $l \in \{1,2 \}$. By Lemma~\ref{upper:bound} below, it follows that the first joint zero $\hat{\bm{z}}$ of $f_1(\tilde{C}_1,\tilde{C}_2;\cdot)$ and $f_2(\tilde{C}_1,\tilde{C}_2;\cdot)$ is bounded by $(\varepsilon/2) {\bf{v}}$. By (\ref{directional:der}) we may assume that ${\bf{v}}$ is such that $v_1+v_2=1$ where ${\bf{v}}=(v_1,v_2)$. It then follows by Theorem~\ref{thm:threshold:model}  that $n^{-1}(\mathcal{S}_{1,2n-1} + \mathcal{S}_{2,2n-1})\leq \varepsilon$.

To show part 2. about non-resilience: Again we start with an uninfected network $(C_1,C_2)$, which implies that $f_1 (C_1,C_2;\bm{0})=f_2 (C_1,C_2;\bm{0})=0$. By assumption, $$\min \{ D_{\bf{v}} f_1 (C_1,C_2;{\bf{0}}), D_{\bf{v}} f_2 (C_1,C_2;{\bf{0}}) \}=:\kappa > 0,$$ and by continuity on the directional derivative in the direction $\bf{v}$, there exists some neighborhood $B_{\delta}(\bm{0})$ of $\bm{0}$ such that 
$$\min \{ D_{\bf{v}} f_1 (C_1,C_2;{\bm{z}}), D_{\bf{v}} f_2 (C_1,C_2;{\bm{z}}) \}=:\kappa/2 > 0,$$ for $\bf{z} \in B_{\delta}(\bm{0})$. This implies that $f_1 (C_1,C_2;t{\bf{v}}) ,f_2 (C_1,C_2;t{\bf{v}})>0$ for $t\leq \delta .$ 

We show that this in fact implies that $\max \{ f_1 (C_1,C_2;{\bf{z}}) ,f_2 (C_1,C_2;{\bf{z}}) \}>0 $ for any $\bm{z} \in (\bf{0}, \delta \bf{v}]$. Let now $\tilde{\bm{z}} \in (\bf{0}, \delta \bf{v}]$ and assume without loss of generality that $\tilde{\bm{z}}$ is to the right of the line $\delta \bf{v}$ in $\mathbb{R}^2$. Let $t_{\tilde{\bm{z}}}\bf{v}$ the unique point on the line with the same first coordinate as the point $\tilde{\bf{z}}$, i.e. $t_{\tilde{\bm{z}}} v_1 =z_1$ where ${\bf{v}}=(v_1, v_2) $. We then know that $f_1 (C_1,C_2;t_{\tilde{{\bm{z}}}}{\bf{v}}) ,f_2 (C_1,C_2;t_{\tilde{\bm{z}}}{\bf{v}})>0$. However, by monotonicity of $f_2 (C_1,C_2;{\bm{z}})$ with respect to $z_1$, it follows that $f_2 (C_1,C_2;{\tilde{\bm{z}}})\geq f_2 (C_1,C_2;t_{\tilde{\bm{z}}}{\bf{v}})>0$. If instead $\tilde{\bm{z}}$ was on the left of the line $\delta \bf{v}$ the same argument applied to the second coordinate and the function $f_1$ would show that $f_1 (C_1,C_2;{\tilde{\bm{z}}})>0$.
 Because $\tilde{\bm{z}} \in (\bf{0}, \delta \bf{v}]$ was chosen arbitrary, we conclude that $\max \{ f_1 (C_1,C_2;{\bf{z}}) ,f_2 (C_1,C_2;{\bf{z}}) \}>0 $ for any $\bm{z} \in (\bf{0}, \delta \bf{v}]$.

Now consider the infected network parameterized by $(\tilde{C}_1,\tilde{C}_2)$ with $\P ((\tilde{C}_1,\tilde{C}_2) \leq (C_1,C_2))=1$ and $\P (\tilde{C}_1 ,\tilde{C}_2)\in \mathbf{D}_{J,1} \cup \mathbf{D}_{J,2} )>0$. It follows by Corollary~\ref{monotonicity:f} item 2. that $$f_l(C_1,C_2;z_1,z_2) < f_l(\tilde{C}_1,\tilde{C}_2;z_1,z_2)$$ for all $z_1,z_2\geq 0$ and $l \in \{ 1,2 \}$.
By above consideration we know that no point $\bm{z} \in (\bf{0}, \delta \bf{v}]$ exists such that both $f_1 (C_1,C_2;\cdot)$ and $f_2 (C_1,C_2;\cdot ) $ are $\leq 0$. This implies that there exists also no point $\bm{z} \in (\bf{0}, \delta \bf{v}]$ such that both $f_1 (\tilde{C}_1,\tilde{C}_2;\cdot)$ and $f_2 (\tilde{C}_1,\tilde{C}_2;\cdot ) $ are $\leq 0$. In particular because of $\max \{ f_1 (\tilde{C}_1,\tilde{C}_2;{\bf{0}}), f_2 (\tilde{C}_1,\tilde{C}_2;{\bf{0}} ) \}>0$, there can not be any joint zero of $f_1 (\tilde{C}_1,\tilde{C}_2;\cdot)$ and $f_2 (\tilde{C}_1,\tilde{C}_2;\cdot ) $ in $[\bf{0}, \delta \bf{v}]$. Again, by Theorem~\ref{thm:threshold:model} part 1. it follows that $n^{-1}(\mathcal{S}_{1,2n-1} + \mathcal{S}_{2,2n-1})\geq \delta v_1 + \delta v_2 $, independent of the specification of $(\tilde{C}_1,\tilde{C}_2)$.
\end{proof}
\begin{proof}[Proof Lemma~\ref{upper:bound}]
First note that $f_1(0,z_2)\geq 0$ for any $z_2\geq 0$ and $f_2(0,z_2)\geq 0$ for any $z_1\geq 0$. We now construct a decreasing sequence of points $\{\bm{z}_n=(z_{1,n},z_{2,n})\}_{n\in \mathbb{N}}$ with $\bm{z}_0=\tilde{\bm{z}}$ and such that $\lim_{n\rightarrow \infty}\bm{z}_n $ exists and is a joint zero of $f_1$ and $f_2$. 

For this, set $\bm{z}_0=\tilde{\bm{z}}$, and then for every $n>0$ odd we define $\bm{z}_n=(z_{1,n},z_{2,n})$ by $z_{2,n}=z_{2,n-1}$ and 
\begin{equation}\label{sup:1}
z_{1,n}=\sup \{t\leq z_{1,n-1} | f_1 (t, z_{2,n-1} )=0\}
\end{equation}
and for even $n$ we define $\bm{z}_n=(z_{1,n},z_{2,n})$ by $z_{1,n}=z_{1,n-1}$ and
\begin{equation}\label{sup:2}
z_{2,n}=\sup \{t\leq z_{2,n-1} | f_2 (z_{1,n-1},t )=0\}.
\end{equation}
We first need to show that the above sets are not empty and that $(z_{1,n}, z_{2,n})$ is actually a well defined point in $[0,\infty) \times [0, \infty )$ for every $n$. By assumption, we have that $f_1(\tilde{\bm{z}}) , f_2(\tilde{\bm{z}})\leq 0$. For the set defined in (\ref{sup:1}) as long as $f_1 (\bm{z}_{n-1})\leq 0$, it follows by continuity of $f_1$ and by $f_1(0,z_2)\geq 0$ that the set is not empty and $z_{1.n} \in [0,\infty) $. Moreover, by definition of $z_{1,n}$, it holds that $f_1(\tilde{\bm{z}}_n)=0$ and by monotonicity of $f_2$ in its first argument it holds that $f_2(\tilde{\bm{z}}_n) \leq 0$ as long as $f_2(\tilde{\bm{z}}_{n-1}) \leq 0$. The same argument applies to (\ref{sup:2}). 

Clearly, the sequence $\bm{z}_n$ is decreasing and bounded from below by $\bm{0}$. This implies that the sequence converges to some point $\overline{\bm{z}} \leq \tilde{\bm{z}}$. Moreover, because $f_1 (\bm{z}_n)=0$ whenever $n$ is odd and $f_2 (\bm{z}_n)=0$ whenever $n$ is even, it follows by continuity of $f$ that $f_1 (\overline{\bm{z}}) = f_2 (\overline{\bm{z}}) =0$ and thus $\overline{\bm{z}}$ is a joint zero. Note that the sequence $\tilde{\bm{z}}_{n}$ is defined to approximate the joint zero $\overline{\bm{z}}$ from a region where both function values are $\leq 0$.
\end{proof}
\bibliographystyle{apalike}
\small{\bibliography{bibtex2}}

\end{document}

%% file: definitions.tex
\usepackage{mathrsfs}

\newcommand{\E}{{\mathbb{E}}}
\newcommand{\dd}{{\rm d}}

\providecommand{\N}{{\mathbb N}}

\newcommand{\1}{\ensuremath{\mathbf{1}}}

\ifx\prop\undefined
\newtheorem{prop}{Proposition}[section]
\fi

\providecommand{\D}{\ensuremath{D(\R_+,\R)}}

\def\P{{\mathbb P}}

\providecommand{\abs}[1]{\ensuremath{\left\lvert#1\right\rvert}}

\providecommand{\poi}{\ensuremath{{\rm Poi}}}

%% file: main.bbl
\begin{thebibliography}{}

\bibitem[Amini et~al., 2016a]{Cont2016}
Amini, H., Cont, R., and Minca, A. (2016a).
\newblock {Resilience to Contagion in Financial Networks}.
\newblock {\em Mathematical Finance}, 26(2):329--365.

\bibitem[Amini et~al., 2013]{AFM13}
Amini, H., Filipovi\'{c}, D., and Minca, A. (2013).
\newblock Systemic risk with central counterparty clearing.
\newblock Swiss {Finance} {Institute} {Research} {Paper} {No.} 13-34, Swiss
  Finance Institute.

\bibitem[Amini et~al., 2016b]{AFM16}
Amini, H., Filipovi\'{c}, D., and Minca, A. (2016b).
\newblock Uniqueness of equilibrium in a payment system with liquidation costs.
\newblock {\em Operations Research Letters}, 44(1):1--5.

\bibitem[Amini and Fountoulakis, 2014]{Amini2014Threshold}
Amini, H. and Fountoulakis, N. (2014).
\newblock {Bootstrap Percolation in Power-Law Random Graphs}.
\newblock {\em Journal of Statistical Physics}, 155(1):72--92.

\bibitem[Amini et~al., 2014]{Amini2014}
Amini, H., Fountoulakis, N., and Panagiotou, K. (2014).
\newblock {Bootstrap percolation in Inhomogeneous random graphs}.
\newblock {\em arXiv:1402.2815}.

\bibitem[Anand et~al., 2015]{ACP14}
Anand, K., Craig, B., and Von~Peter, G. (2015).
\newblock Filling in the blanks: Network structure and interbank contagion.
\newblock {\em Quantitative Finance}, 15(4):625--636.

\bibitem[Arinaminpathy et~al., 2012]{arinaminpathy2012size}
Arinaminpathy, N., Kapadia, S., and May, R.~M. (2012).
\newblock Size and complexity in model financial systems.
\newblock {\em Proceedings of the National Academy of Sciences},
  109(45):18338--18343.

\bibitem[Balogh and Pittel, 2007]{MR2283230}
Balogh, J. and Pittel, B.~G. (2007).
\newblock Bootstrap percolation on the random regular graph.
\newblock {\em Random Structures Algorithms}, 30(1-2):257--286.

\bibitem[Bardoscia et~al., 2017]{BBCH17}
Bardoscia, M., Barucca, P., Brinley~Codd, A., and Hill, J. (2017).
\newblock The decline of solvency contagion risk.
\newblock {\em Bank of England Staff Working Paper}, 662.

\bibitem[Bichuch and Chen, 2020]{bichuch2020systemic}
Bichuch, M. and Chen, K. (2020).
\newblock Systemic risk: the effect of market confidence.
\newblock {\em International Journal of Theoretical and Applied Finance},
  23(07):2050043.

\bibitem[Bichuch and Feinstein, 2019]{bichuch2019optimization}
Bichuch, M. and Feinstein, Z. (2019).
\newblock Optimization of fire sales and borrowing in systemic risk.
\newblock {\em SIAM Journal on Financial Mathematics}, 10(1):68--88.

\bibitem[Bichuch and Feinstein, 2022]{bichuch2022repo}
Bichuch, M. and Feinstein, Z. (2022).
\newblock A repo model of fire sales with vwap and lob pricing mechanisms.
\newblock {\em European Journal of Operational Research}, 296(1):353--367.

\bibitem[Boss et~al., 2004]{BEST04}
Boss, M., Elsinger, H., Summer, M., and Thurner, S. (2004).
\newblock Network topology of the interbank market.
\newblock {\em Quantitative Finance}, 4(6):677--684.

\bibitem[Capponi et~al., 2016]{CCY16}
Capponi, A., Chen, P.-C., and Yao, D.~D. (2016).
\newblock Liability concentration and systemic losses in financial networks.
\newblock {\em Operations Research}, 64(5):1121--1134.

\bibitem[Chen et~al., 2016]{CLY14}
Chen, N., Liu, X., and Yao, D.~D. (2016).
\newblock An optimization view of financial systemic risk modeling: The network
  effect and the market liquidity effect.
\newblock {\em Operations Research}, 64(5).

\bibitem[Cifuentes et~al., 2005]{CFS05}
Cifuentes, R., Shin, H.~S., and Ferrucci, G. (2005).
\newblock Liquidity risk and contagion.
\newblock {\em Journal of the European Economic Association}, 3(2-3):556--566.

\bibitem[Detering and Lin, 2022]{https://doi.org/10.48550/arxiv.2205.14782}
Detering, N. and Lin, J. (2022).
\newblock Percolation in random graphs of unbounded rank.
\newblock {\em arXiv}.

\bibitem[{Detering} et~al., 2019]{Detering2015a}
{Detering}, N., {Meyer-Brandis}, T., and {Panagiotou}, K. (2019).
\newblock {Bootstrap Percolation in Directed Inhomogeneous Random Graphs}.
\newblock {\em The Electronic Journal of Combinatorics}, 26(3).

\bibitem[Detering et~al., 2019]{detering2019managing}
Detering, N., Meyer-Brandis, T., Panagiotou, K., and Ritter, D. (2019).
\newblock Managing default contagion in inhomogeneous financial networks.
\newblock {\em SIAM Journal on Financial Mathematics}, 10(2):578--614.

\bibitem[Detering et~al., 2020]{detering2020financial}
Detering, N., Meyer-Brandis, T., Panagiotou, K., and Ritter, D. (2020).
\newblock Financial contagion in a stochastic block model.
\newblock {\em International Journal of Theoretical and Applied Finance},
  23(08):1--53.

\bibitem[Detering et~al., 2021]{deteringDefaultFireSales}
Detering, N., Meyer-Brandis, T., Panagiotou, K., and Ritter, D. (2021).
\newblock An integrated model for fire sales and default contagion.
\newblock {\em Mathematics and Financial Economics}, 15(1):59--101.

\bibitem[Detering et~al., 2022]{detering2020suffocating}
Detering, N., Meyer-Brandis, T., Panagiotou, K., and Ritter, D. (2022).
\newblock Suffocating fire sales.
\newblock {\em SIAM Journal on Financial Mathematics}, 13(1):70--108.

\bibitem[Eisenberg and Noe, 2001]{EN01}
Eisenberg, L. and Noe, T.~H. (2001).
\newblock Systemic risk in financial systems.
\newblock {\em Management Science}, 47(2):236--249.

\bibitem[Elliott et~al., 2014]{EGJ14}
Elliott, M., Golub, B., and Jackson, M.~O. (2014).
\newblock Financial networks and contagion.
\newblock {\em American Economic Review}, 104(10):3115--3153.

\bibitem[Elsinger, 2009]{E07}
Elsinger, H. (2009).
\newblock Financial networks, cross holdings, and limited liability.
\newblock {\em {\"{O}sterrei}chische Nationalbank (Austrian Central Bank)},
  156.

\bibitem[Elsinger et~al., 2013]{ELS13}
Elsinger, H., Lehar, A., and Summer, M. (2013).
\newblock Network models and systemic risk assessment.
\newblock In {\em Handbook on Systemic Risk}, pages 287--305. Cambridge
  University Press.

\bibitem[Gai et~al., 2011]{G11}
Gai, P., Haldane, A., and Kapadia, S. (2011).
\newblock Complexity, concentration and contagion.
\newblock {\em Journal of Monetary Economics}, 58(5):453--470.

\bibitem[Gai and Kapadia, 2010]{ar:gk10}
Gai, P. and Kapadia, S. (2010).
\newblock {Contagion in Financial Networks}.
\newblock {\em Proc. R. Soc. A}, page 2401–2423.

\bibitem[Glasserman and Young, 2015]{GY14}
Glasserman, P. and Young, H.~P. (2015).
\newblock How likely is contagion in financial networks?
\newblock {\em Journal of Banking and Finance}, 50:383--399.

\bibitem[Gouri\'{e}roux et~al., 2012]{GHM12}
Gouri\'{e}roux, C., H\'{e}am, J.-C., and Monfort, A. (2012).
\newblock Bilateral exposures and systemic solvency risk.
\newblock {\em Canadian Journal of Economics}, 45(4):1273--1309.

\bibitem[Halaj and Kok, 2015]{HK15}
Halaj, G. and Kok, C. (2015).
\newblock Modelling the emergence of the interbank networks.
\newblock {\em Quantitative Finance}, 15(4):653--671.

\bibitem[Hofstad, 2016]{Hofstad2014}
Hofstad, R. v.~d. (2016).
\newblock {\em {Random Graphs and Complex Networks}}, volume~1 of {\em
  {Cambridge Series in Statistical and Probabilistic Mathematics}}.
\newblock Cambridge University Press.

\bibitem[H\"{u}ser, 2015]{H16}
H\"{u}ser, A.-C. (2015).
\newblock Too interconnected to fail: A survey of the interbank networks
  literature.
\newblock {\em Journal of Network Theory in Finance}, 1(3):1--50.

\bibitem[Janson et~al., 2012]{janson2012}
Janson, S., {\L}uczak, T., Turova, T., and Vallier, T. (2012).
\newblock {Bootstrap percolation on the random graph $G_{n,p}$}.
\newblock {\em Ann. Appl. Probab.}, 22(5):1989--2047.

\bibitem[Lewis, 2010]{Lewis2010big}
Lewis, M. (2010).
\newblock {\em The Big Short: Inside the Doomsday Machine}.
\newblock New York: Simon \& Schuster.

\bibitem[Nier et~al., 2007]{NYYA07}
Nier, E., Yang, J., Yorulmazer, T., and Alentorn, A. (2007).
\newblock Network models and financial stability.
\newblock {\em Journal of Economic Dynamics and Control}, 31(6):2033--2060.

\bibitem[Rogers and Veraart, 2013]{RV13}
Rogers, L.~C. and Veraart, L.~A. (2013).
\newblock Failure and rescue in an interbank network.
\newblock {\em Management Science}, 59(4):882--898.

\bibitem[Staum, 2013]{Staum}
Staum, J. (2013).
\newblock Counterparty contagion in context: Contributions to systemic risk.
\newblock In {\em Handbook on Systemic Risk}, pages 512--548. Cambridge
  University Press.

\bibitem[Upper, 2011]{U11}
Upper, C. (2011).
\newblock Simulation methods to assess the danger of contagion in interbank
  markets.
\newblock {\em Journal of Financial Stability}, 7(3):111--125.

\bibitem[Weber and Weske, 2017]{AW_15}
Weber, S. and Weske, K. (2017).
\newblock The joint impact of bankruptcy costs, fire sales and cross-holdings
  on systemic risk in financial networks.
\newblock {\em Probability, Uncertainty and Quantitative Risk}, 2(1):9.

\bibitem[Wormald, 1995]{10.1214/aoap/1177004612}
Wormald, N.~C. (1995).
\newblock {Differential Equations for Random Processes and Random Graphs}.
\newblock {\em The Annals of Applied Probability}, 5(4):1217 -- 1235.

\end{thebibliography}
